\documentclass[a4paper,english,cleveref]{lipics-v2021}

\usepackage{graphicx}
\usepackage{amssymb,amsfonts,amsmath}
\usepackage[utf8]{inputenc}
\usepackage[T1]{fontenc}
\usepackage{complexity}
\usepackage{xspace}
\usepackage{booktabs}
\usepackage{nicefrac}

\title{Adjacency Graphs of Polyhedral Surfaces}

\author{Elena~Arseneva}{Universit\`{a} della Svizzera italiana, Switzerland}{elena.arseneva@usi.ch}{https://orcid.org/0000-0002-5267-4512}{}

\author{Linda~Kleist}{Technische Universit\"at Braunschweig, Germany}{kleist@ibr.cs.tu-bs.de}{https://orcid.org/0000-0002-3786-916X}{}

\author{Boris~Klemz}{University W\"urzburg, Germany}{firstname.lastname ``at'' uni-wuerzburg.de}{https://orcid.org/0000-0002-4532-3765}{partially supported by DFG project WO$\,$758/11-1.}

\author{Maarten~L\"offler}{Utrecht University, the Netherlands}{m.loffler@uu.nl}{}{}

\author{Andr\'e~Schulz}{FernUniversit\"at in Hagen, Germany}{andre.schulz@fernuni-hagen.de}{https://orcid.org/0000-0002-2134-4852}{}

\author{Birgit~Vogtenhuber}{Technische Universit\"at Graz, Austria}{bvogt@ist.tugraz.at}{https://orcid.org/0000-0002-7166-4467}{partially supported by the Austrian
  Science Fund within the collaborative DACH project
  \emph{Arrangements and Drawings} as FWF project \mbox{I 3340-N35}.}

\author{Alexander Wolff}{University W\"urzburg,  Germany}{}{https://orcid.org/0000-0001-5872-718X}{}

\authorrunning{E.~Arseneva, L.~Kleist, B.~Klemz, M.~L\"offler,
  A.~Schulz, B.~Vogtenhuber, and A.~Wolff}

\Copyright{Elena Arseneva, Linda Kleist, Boris Klemz, Maarten
  L\"offler,\\Andr\'e Schulz, Birgit Vogtenhuber, and Alexander Wolff}



\acknowledgements{We thank the organizers of Dagstuhl Seminar 19352
  ``Computation in Low-Dimensional Geometry and Topology'' for
  bringing us together.  We are particularly indebted to seminar
  participant Arnaud de Mesmay for asking a question that initiated
  our research.  We also thank Louis Esperet for his pointer to
  reference~\cite{thomassenColorCriticalGraphs} and the connection to
  an open problem concerning the chromatic number of adjacency graphs.
  Last but not least, we thank the anonymous referees of our EuroCG
  2020 and SoCG 2021 submission for their helpful comments.  }

\relatedversion{A preliminary version of this article appeared in the
  proceedings of the \emph{International Symposium on Computational
    Geometry} (SoCG) 2021 \cite{DBLP:conf/compgeom/ArsenevaKKL0V021}.}

\Crefname{figure}{Fig.}{Figs.}
\crefname{figure}{Figure}{Figures}

\newtheorem{question}[theorem]{Question}

\newcommand{\bigO}{\mathcal{O}}
\graphicspath{{./figures/}}

\newcommand{\Sur}{\ensuremath{\mathcal{S}}\xspace}
\newcommand{\Gra}{\ensuremath{\mathcal{G}}\xspace}
\newcommand{\TheThreeTree}{triple-stacked triangle\xspace}
\newcommand{\com}{\ensuremath{\zeta}\xspace}

\hideLIPIcs
\nolinenumbers

\begin{document}

\maketitle

\begin{abstract}
  We study whether a given graph can be realized as an adjacency graph
  of the polygonal cells of a polyhedral surface in~$\mathbb{R}^3$.
  We show that every graph is realizable as a polyhedral surface with
  arbitrary polygonal cells, and that this is not true if we require
  the cells to be convex.  In particular, if the given graph contains
  $K_5$, $K_{5,81}$, or any nonplanar $3$-tree as a subgraph, no such
  realization exists.  On the other hand, all planar graphs,
  $K_{4,4}$, and $K_{3,5}$ can be realized with convex cells.  The
  same holds for any subdivision of any graph where each edge is
  subdivided at least once, and, by a result from McMullen et
  al.~(1983), for any hypercube.

  Our results have implications on the maximum density of graphs
  describing polyhedral surfaces with convex cells: The realizability
  of hypercubes shows that the maximum number of edges over all
  realizable $n$-vertex graphs is in $\Omega(n \log n)$.  From the
  non-realizability of $K_{5,81}$, we obtain that any realizable
  $n$-vertex graph has $\bigO(n^{9/5})$ edges.  As such, these graphs
  can be considerably denser than planar graphs, but not arbitrarily
  dense.
\end{abstract}

\section{Introduction}

A {\em polyhedral surface} consists of a set of interior-disjoint
polygons embedded in $\mathbb{R}^3$, where each edge may be shared by
at most two polygons.  Polyhedral surfaces have been long studied in
computational geometry, and have well-established applications in for
instance computer graphics~\cite{Dobkin92computationalgeometry} and
geographical information science~\cite{fmp-acggi-97}.

Inspired by those applications, classic work in this area often
focuses on restricted cases, such as surfaces of (genus 0)
polyhedra~\cite {DBLP:journals/dcg/AronovKOV03,10.1145/276884.276901},
or $x,y$-monotone surfaces known as {\em polyhedral
  terrains}~\cite{COLE198911}.  Such surfaces are, in a sense,
2-dimensional. One elegant way to capture this ``essentially
2-dimensional behaviour'' is to look at the adjacency graph (see below
for a precise definition) of the surface: in both cases described
above, this graph is \emph{planar}.  In fact, by Steinitz's Theorem
the adjacency graphs of surfaces of convex polyhedra are exactly the
3-connected planar graphs~\cite{S22}.  If we allow the surface of a
polyhedron to have a {\em boundary}, then every planar graph has a
representation as such a polyhedral surface~\cite{DGHKK12}.

Recently, applications in computational topology have intensified the
study of polyhedral surfaces of non-trivial topology.  In sharp
contrast to the simpler case above, where the classification is
completely understood, little is known about the class of adjacency
graphs that describe general polyhedral surfaces.  In this paper we
investigate this graph class.
  
\subparagraph{Our model.}

A polyhedral surface $\Sur = \{S_1, \ldots, S_n\}$ is a set of $n$
closed polygons embedded in $\mathbb{R}^3$ such that, for all pairwise
distinct indices $i,j,k\in \{1,2,\dots,n\}$:
\begin {itemize}
\item $S_i$ and $S_j$ are interior-disjoint (with respect to the 2D
  relative interior of the objects);
\item if $S_i \cap S_j \neq \emptyset$, then $S_i \cap S_j$ is either
  a single %
  corner or a complete {\em side} of both~$S_i$ and~$S_j$;
\item if $S_i \cap S_j \cap S_k \neq \emptyset$ then it is a single
  corner (i.e., a side is shared by at most two polygons).
\end {itemize}
To avoid confusion with the corresponding graph elements, we
consistently refer to polygon vertices as \emph{corners} and to
polygon edges as \emph{sides}.
  
The \emph{adjacency graph} of a polyhedral surface \Sur, denoted as
$\Gra(\Sur)$, is the graph whose vertices correspond to the polygons
of \Sur and which has an edge between two vertices if and only if the
corresponding polygons of \Sur share a side. Note that a
corner--corner contact is allowed in our model but does not induce an
edge in the adjacency graph. Further observe that the adjacency graph
does not uniquely determine the topology of the surface.
\Cref{fig:3deg} shows an example of a polyhedral surface and its
adjacency graph.  We say that a polyhedral surface \Sur
\emph{realizes} a graph $G$ if $\Gra(\Sur)$ is isomorphic to $G$.  In
this case, we write $\Gra(\Sur) \simeq G$.
  
\begin{figure}[htb]
  \centering
  \begin{subfigure}[b]{0.23\textwidth}
    \centering \includegraphics[page=3,scale=.25]{3-tree-possible}
    \caption{a surface \Sur.}
    \label{fig:3degB}
  \end{subfigure}
  \hfil
  \begin{subfigure}[b]{0.23\textwidth}
    \centering \includegraphics[page=1]{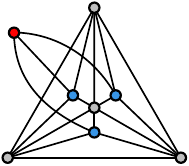}
    \caption{the graph $\Gra(\Sur)$.}
    \label{fig:3degA}
  \end{subfigure}
  \caption{A convex-polyhedral surface \Sur and its nonplanar
    3-degenerate adjacency graph $\Gra(\Sur)$.}
  \label{fig:3deg}
\end{figure}

If every polygon of a polyhedral surface \Sur is strictly convex, we
call \Sur a \emph{convex-polyhedral} surface.  Our paper focuses on
convex-polyhedral surfaces; refer to \Cref{fig:non-convex} for an
example of a general (nonconvex) polyhedral surface.  We emphasize
that we do not require that every polygon side has to be shared with
another polygon.

Our work relates to two lines of research: Steinitz-type problems and
contact representations.

\subparagraph{Steinitz-type problems.}

Steinitz's Theorem gives the positive answer to the
\emph{realizability problem} for convex polyhedra.  This result is
typically stated in terms of the realizabilty of a graph as the
1-skeleton of a convex polyhedron.  Our perspective comes from the
dual point of view, describing the adjacencies of the faces instead of
the adjacencies of the vertices.

Steinitz's Theorem settles the problem raised in this paper for
surfaces that are homeomorphic to a sphere.  A slightly stronger
version of Steinitz's Theorem by Gr\"unbaum and Barnette~\cite{BG69}
states that every planar 3-connected graph can be realized as the
1-skeleton of a convex polyhedron with the prescribed shape of one
face.  Consequently, also in our model we can prescribe the shape of
one polygon if the adjacency graph of the surface is planar.  For
other classes of polyhedra only very few partial results for their
graph-theoretic characterizations are known~\cite{EM14,HN11}.  No
generalization for Steinitz's Theorem for surfaces of higher genus is
known, and therefore there are also no results for the dual
perspective.  In higher dimensions, Richter-Gebert's Universality
Theorem implies that the realizability problem for abstract
4-polytopes is
$\exists \mathbb{R}$-complete~\cite{richterGebertrealization}.

McMullen et al.\ \cite{msw-p2me3-IJM83} constructed a closed
polyhedral surface of genus 4097~-- with only 4096 polygons.  For the
first few steps of their inductive construction; see
\Cref{fig:mcmullen}.  Their construction answered a question that
Barnette posed in 1980; he asked whether there are polyhedra in
$\mathbb{R}^3$ whose polygonal faces all have arbitrarily many sides.
Later, Ziegler \cite{z-pshg-DDG08} gave a different construction of
the family of surfaces presented by McMullen et al.

\subparagraph{Simplicial complexes.}

The algorithmic problem of determining whether a given
$k$-dimen\-sional (abstract) simplicial complex embeds in
$\mathbb R^d$ is an active field of
research~\cite{vcadek2017algorithmic,undecidable2020,matousek20011hardness,mesmay2020embeddability,Skopenkov2020,skopenkov2020invariants}.
There exist at least three interesting notions of embeddability:
linear, piecewise linear, and topological embeddability, which usually
are not the same~\cite{matousek20011hardness}.  The case
$(k,d)=(1,2)$, however, corresponds to testing graph planarity, and
thus, all three notions coincide, and the problem lies in~\P.

While some necessary conditions for the geometric \emph{realizability
  of simplicial complexes} are
known~\cite{novik2000note,timmreck2008necessary}, the problem of
recognizing the linear embeddability of $k$-dimensional complexes into
$\mathbb{R}^d$ is conjectured to be \NP-hard for every fixed pair
$(k,d)$ with $3 \leq d \leq 3k + 1$~\cite[Conjecture
3.2.2]{skopenkov2020invariants}.  Recently, Abrahamsen, Kleist, and
Miltzow~\cite{GeometricEmbeddings2021} showed that deciding whether a
2-simplex (e.g., a set of triangles with prescribed edge contacts),
linearly embeds in $\mathbb R^3$ is $\exists\mathbb R$-complete; this
remains true even if a piecewise linear embedding is given. More
generally, they showed $\exists\mathbb R$-completeness for the
decision problem of linearly embedding a $k$-simplex in $\mathbb R^d$
for all $d\geq 3$ and $k\in\{d,d-1\}$.

Concerning piecewise-linear embeddability, determining whether a given
$k$-complex embeds piecewise-linearly in $\mathbb R^d$ for the cases
$d=3$ and $k \in \{2,3\}$ is known to be
\NP-hard~\cite{mesmay2020embeddability} and
decidable~\cite{embeddibility3sphere}.  In higher dimensions, the
problem is polynomial time solvable for $d\geq 4 $ and
$k < \nicefrac{2}{3}\cdot(d-1)$~\cite{vcadek2017algorithmic}, \NP-hard
for $d\geq 4,k \geq \nicefrac{2}{3}\cdot(d-1)$ and even undecidable
for $d \geq 5$ and
$k\in\{d,d-1\}$~\cite{undecidable2020,Skopenkov2020}.

\subparagraph{Contact representations.}

A realization of a graph as a polyhedral surface can be viewed as a
\emph{contact representation} of this graph with polygons in
$\mathbb R^3$, where a contact between two polygons is realized by
sharing an entire polygon side, and each side is shared by at most two
polygons.  In a general contact representation of a graph, the
vertices are represented by interior-disjoint geometric objects, where
two objects touch if and only if the corresponding vertices are
adjacent.  In concrete settings, the object type (disks, lines,
polygons, etc.), the type of contact, and the embedding space is
specified.  Numerous results concerning which graphs admit a contact
representation of some type are known; we review some of them.

The well-known Andreev--Koebe--Thurston circle packing
theorem~\cite{a-cpls-70,Koebe36} states that every planar graph admits
a contact representation by touching disks in~$\mathbb{R}^2$.  A less
known but impactful generalization by
Schramm~\cite[Theorem~8.3]{schramm} guarantees that every
triangulation (i.e., maximal planar graph) has a contact
representation in $\mathbb{R}^2$ where every inner vertex corresponds
to a homothetic copy of a prescribed smooth convex set; the three
outer vertices correspond to prescribed smooth arcs whose union is a
simple closed curve.  If the prototypes and the curve are polygonal,
i.e., are not smooth, then there still exists a contact
representation, however, with the following shortcomings: The sets
representing inner vertices may degenerate to points, which may lead
to extra contacts.  As observed by Gon{\c{c}}alves et
al.~\cite{JGAA-509}, Schramm's result implies that every subgraph of a
4-connected triangulation has a contact representation with aligned
equilateral triangles and similarly, every inner triangulation of a
4-gon without separating 3- and 4-cycles has a hole-free contact
representation with squares~\cite{s-stpc-IJM84,Fel13}.

While for the afore-mentioned existence results there are only
iterative procedures that compute a series of representations
converging to the desired one, there also exist a variety of shapes
for which contact representations can be computed efficiently.
Allowing for sides of one polygon to be \emph{contained} in the side
of adjacent polygons, Duncan et al.~\cite{DGHKK12} showed that, in
this model, every planar graph can be realized by hexagons in the
plane and that hexagons are sometimes necessary.  Assuming
side--corner contacts, de Fraysseix et al.~\cite{FMR94} showed that
every plane graph has a triangle contact representation and how to
compute one.  Gansner et al.~\cite{GHK10a} presented linear-time
algorithms for triangle side-contact representations for outerplanar
graphs, square grid graphs, and hexagonal grid graphs.  Kobourov et
al.~\cite{KMN13} showed that every 3-connected cubic planar graph
admits a triangle side-contact representation whose triangles form a
tiling of a triangle.  For a survey of planar graphs that can be
represented by dissections of a rectangle into rectangles, we refer to
Felsner~\cite{Fel13}.  Moreover, there exist linear-time algorithms to
compute hole-free contact representations of triangulations where each
vertex is represented by a 8-sided rectilinear
polygon~\cite{de1994triangle,doi:10.1137/S0097539796308874,doi:10.1137/0222035}.
In fact, Alam et al.~\cite{alam2013computing} showed that there exist
contact representations where the area of the polygons can even be
prescribed (however, no polynomial-time algorithm is known to compute
such representations).  On the negative side, Breu and
Kirkpatrick~\cite{DBLP:journals/comgeo/BreuK98} showed that
recognizing whether a given graph admits a contact representation with
\emph{unit} disks is \NP-hard.  Later, Klemz et
al.~\cite{DBLP:journals/jocg/KlemzNP22} showed that this statement
remains true even when restricted to outerplanar graphs.  Moreover,
Bowen et al.~\cite{DBLP:conf/gd/BowenDLR0T15} showed that if the unit
disk contact representation is additionally required to respect a
given rotation system, the recognition problem is \NP-hard even when
restricted to trees.

Representations with one-dimensional objects in $\mathbb{R}^2$ have
also been studied.  While every plane bipartite graph has a contact
representation with horizontal and vertical segments~\cite{FMP91},
Hlin\v{e}n\'{y}~\cite{HLINENY200195} showed that recognizing segment
contact graphs is an \NP-complete problem even when restricted to
planar graphs.  Hlin\v{e}n\'{y}~\cite{Hli98} also showed that
recognizing curve contact graphs where no four curves meet in one
point is \NP-complete for planar graphs whereas the same question can
be solved in polynomial time for planar triangulations.

Less is known about contact representations in higher dimensions.
Every graph is the contact graph of interior-disjoint convex polytopes
in $\mathbb{R}^3$ where contacts are shared 2-dimensional
facets~\cite{tietze1905}.  Hlin\v{e}n\'{y} and
Kratochv{\'\i}l~\cite{HK01} proved that the recognition of unit-ball
contact graphs in $\mathbb{R}^d$ is \NP-hard for $d=3,4,$ and~8.
Felsner and Francis~\cite{FF11} showed that every planar graph has a
contact representation with axis-parallel cubes in~$\mathbb{R}^3$.
For proper side contacts, Kleist and Rahman~\cite{kleist14} proved
that every subgraph of an Archimedean grid can be represented with
unit cubes, and every subgraph of a $d$-dimensional grid can be
represented with $d$-cubes.  Evans et al.~\cite{erssw-rghtp-GD19}
showed that every graph has a contact representation where vertices
are represented by convex polygons in $\mathbb{R}^3$ and edges by
shared corners of polygons, %
and gave polynomial-volume representations for bipartite, 1-planar,
and cubic graphs.

\subparagraph{Contribution and organization.}

We show that for every graph~$G$ there exists a polyhedral
surface~\Sur such that $G$ is the adjacency graph of~\Sur; see
\Cref{sec:general}.  For convex-polyhedral surfaces, the situation is
more intricate; see \Cref{sec:convex}.  Every planar graph can be
realized by a \emph{flat} convex-polyhedral surface
(\Cref{planar:2d}), i.e., a convex-polyhedral surface
in~$\mathbb R^2$.  Some nonplanar graphs cannot be realized by
convex-polyhedral surfaces in~$\mathbb{R}^3$; in particular this holds
for all supergraphs of~$K_5$ (\Cref{obs:K5}), of~$K_{5,81}$
(\Cref{thm:k581}), and of all nonplanar $3$-trees
(\Cref{thm:threeTree}).  Nevertheless, many nonplanar graphs,
including $K_{4,4}$ and $K_{3,5}$, have such a realization
(\Cref{prop:K44,prop:K35}).  We remark that all our positive results
hold for subgraphs and subdivisions as well
(\cref{prop:subgraphsSubdivisions}).  Similarly, our negative results
carry over to supergraphs.

Our results have implications on the maximum density of adjacency
graphs of convex-polyhedral surfaces; see \Cref{sec:density-bounds}.
On the one hand, the non-realizability of $K_{5,81}$ implies that the
number of edges of any realizable $n$-vertex graph is upperbounded by
$\bigO(n^{9/5})$ edges.  On the other hand, the realizability of
hypercubes (which we derive from the above-mentioned result of
McMullen et al.~\cite{msw-p2me3-IJM83}; see \Cref{sec:hypercubes})
implies that there are realizable graphs with $n$ vertices and
$\Omega(n \log n)$ edges.  Hence these graphs can be considerably
denser than planar graphs, but not arbitrarily dense.

\section{General Polyhedral Surfaces}
\label{sec:general}

We start with a simple positive result.

\begin{proposition}
  \label{prop:nonconvex}
  For every graph $G$, there exists a polyhedral surface~\Sur such
  that $\Gra(\Sur) \simeq G$.
\end{proposition}

\begin{proof}
  We start our construction with~$n=|V(G)|$ interior-disjoint
  rectangles such that there is a line segment~$s$ that acts as a
  common side of all these rectangles.  We then cut away parts of each
  rectangle thereby turning it into a comb-shaped polygon as
  illustrated in \Cref{fig:non-convex}.  These polygons represent the
  vertices of~$G$.  For each pair~$(P,P')$ of polygons that are
  adjacent in~$G$, there is a subsegment~$s_{PP'}$ of~$s$ such
  that~$s_{PP'}$ is a side of both~$P$ and~$P'$ that is disjoint from
  the remaining polygons.  In particular, every polygon side is
  adjacent to at most two polygons.  The result is a polyhedral
  surface whose adjacency graph is~$G$.
\end{proof}

\begin{figure}[htb]
  \centering \includegraphics[page=3]{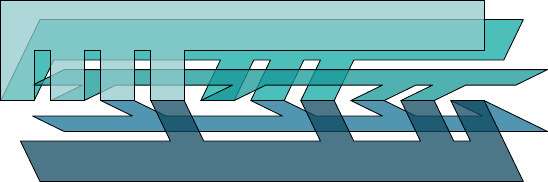}
  \caption{A realization of~$K_5$ by arbitrary polygons with side
    contacts in~$\mathbb{R}^3$.}
  \label{fig:non-convex}
\end{figure}

In our construction, the complexity of each polygon depends on the
\emph{degree} of the vertex it represents. If we insist on strictly
convex polygons and full side contacts, this is clearly also
necessary.  One interesting question is how tight this dependence is.

To make this question precise, for a polyhedral surface $\Sur$, let
$\com(\Sur)$ be the total complexity of $\Sur$; that is, the sum of
the number of vertices (or edges, which is the same) of all the
polygons in $\Sur$.  Then, for a graph $G$, define
\[\displaystyle \com(G) = \min_{\Sur : \Gra(\Sur) \simeq G}
  \com(\Sur)\] to be the complexity of the best possible
representation of $G$.  Proposition~\ref {prop:nonconvex} implies an
upper bound on $\com(G)$.

\begin{corollary}
  \label{cor:3d}
  Let $G$ be a graph with $m=|E(G)|$ edges.  Then $\com(G) \le 6m$.
\end{corollary}

\begin {proof}
  Our construction for Proposition~\ref {prop:nonconvex} represents a
  vertex of degree $d$ by a polygon with at most $4d+2$ sides. It is
  not hard to improve this to exactly $3d$: instead of a rectangular
  comb, we can use a triangular one, with triangular gaps between
  successive teeth.  The corollary now follows from the fact that the
  sum of degrees is twice the number of edges in~$G$.
\end {proof}

For some graphs, this bound is tight; for example, the graph which
consists of a single edge.  However, some graphs admit much better
embeddings.  The lower bound on the number of sides for a vertex of
degree $d$ is exactly $d$. There are graphs that realize this lower
bound: they are exactly the adjacency graphs of \emph {closed}
polyhedral surfaces.  In this model, $K_7$ can be realized as the
so-called Szilassi polyhedron; for an illustration,
see~\cite{wikipedia-szilassi}.  The tetrahedron and the Szilassi
polyhedron are the only two known polyhedra in which each face shares
a side with every other face~\cite{wikipedia-szilassi}.  Which other
(complete) graphs can be realized in this way remains an open problem.

\section{Convex-Polyhedral Surfaces}
\label{sec:convex}

In this section we investigate which graphs can be realized by
\emph{convex}-polyhedral surfaces.  First of all, it is always
possible to represent a subgraph or a subdivision of an adjacency
graph with slight modifications of the corresponding surface:
\emph{trimming} the polygons allows us to represent subgraphs, while
trimming and inserting chains of polygons allows subdivisions.
Consequently, we obtain the following result.

\begin{proposition}
  \label{prop:subgraphsSubdivisions}
  The set of adjacency graphs of convex-polyhedral surfaces
  in~$\mathbb{R}^3$ is closed under taking subgraphs and subdivisions.
\end{proposition}

\begin{proof}
  Obviously the set of adjacency graphs of convex-polyhedral surfaces
  is closed under vertex deletions.  It remains to show that it is
  also closed under edge deletions and edge subdivisions.  Consider a
  surface~$\Sur$ and its adjacency graph~$G$.  We define three
  operations that locally change~$\Sur$ and describe their effect
  on~$G$.

  \begin{figure}[htb]
    \begin{subfigure}[b]{0.49\textwidth}
      \centering \includegraphics[page=5]{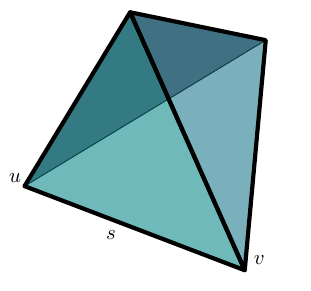}
      \caption{a surface $\mathcal S$ with $\Gra(\Sur) \simeq K_4$.}
      \label{subfig:intial}
    \end{subfigure}
    \hfill
    \begin{subfigure}[b]{0.49\textwidth}
      \centering \includegraphics[page=6]{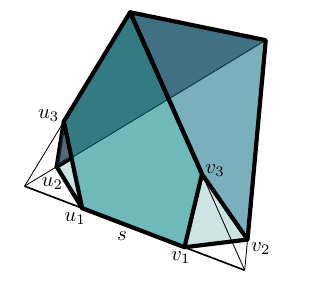}
      \caption{trimming the corners $u$ and $v$ of $\mathcal S$.}
      \label{subfig:corner_trim}
    \end{subfigure}
    
    \begin{subfigure}[b]{0.49\textwidth}
      \centering \includegraphics[page=7]{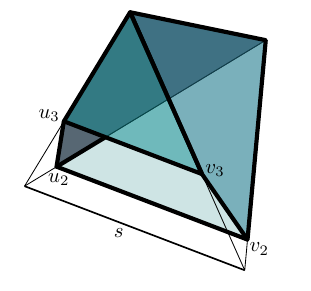}
      \caption{trimming the side $s$ of $\mathcal S$.}
      \label{subfig:side_trim}
    \end{subfigure}
    \hfill
    \begin{subfigure}[b]{0.49\textwidth}
      \centering \includegraphics[page=8]{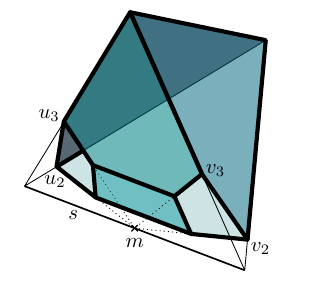}
      \caption{subdividing the side $s$ of $\mathcal S$.}
      \label{subfig:side_subdivide}
    \end{subfigure}
    \caption{The three operations used in the proof of
      \cref{prop:subgraphsSubdivisions}.}
  \end{figure}

  A \emph{corner trim} takes any corner $v$ of $\Sur$ (which may
  belong to multiple polygons of $\Sur$) and replaces it by a set of
  new corners, one on each incident side, all at distance
  $\varepsilon$ from $v$ for some sufficiently small $\varepsilon$;
  see Figures~\ref{subfig:intial} and~\ref{subfig:corner_trim}.  Each
  polygon in $\Sur$ incident to $v$ now uses two new copies of $v$
  instead of $v$; observe that the new polygons are still strictly
  convex.  The adjacency graph of $\Sur$ does not change under a
  corner trim operation.

  A \emph{side trim} takes any side $s$ of $\Sur$ and first performs a
  corner trim on both incident corners.  This creates two new corners
  on $s$; we delete these two new corners from the (at most two)
  polygons incident to $s$; see Figures~\ref{subfig:intial}
  and~\ref{subfig:side_trim}. Note that this operation still preserves
  strict convexity of any polygons incident to $s$, and that if there
  were two polygons that shared~$s$, they now no longer share a side.
  Thus, the edge of $G$ corresponding to $s$ is removed.

  Finally, a \emph{subdivide} operation takes any side $s$ of $\Sur$
  that is incident to two polygons $P$ and $Q$, and first performs a
  side trim on $s$.  Next, we create a new polygon $R$ as follows; for
  an illustration see \Cref{subfig:intial,subfig:side_subdivide}.
  Let~$m$ be the midpoint of $s$.  We first add~$m$ to the trimmed
  versions of both $P$ and $Q$.  Now, we create two lines $\ell_P$ and
  $\ell_Q$ parallel to~$s$ that lie on the supporting planes of $P$
  and $Q$ on the side of $s$ that contains $P$ or $Q$, respectively.
  The distance $\delta$ of these lines to $s$ is chosen sufficiently
  small; in particular, we must take $\delta < \varepsilon$.

  We place two new corners on the intersection of $\ell_P$ with the
  boundary of $P$ and two new corners on the intersection of $\ell_Q$
  with the boundary of $Q$. Note that the resulting points are
  coplanar and in convex position; our new polygon $R$ is the convex
  hull of these points.  If~$\delta$ is chosen sufficiently small, the
  polygon $R$ has a nonempty intersection with $P$ and $Q$, but no
  other polygon of $\Sur$.  Hence, the subdivide operation creates a
  new vertex in $G$ that is adjacent to the two original endpoints of
  the edge that corresponds to $s$, and no other vertices of~$G$.
\end {proof}

The existence of a flat surface with the correct adjacencies follows
from the Andreev--Koebe--Thurston circle packing theorem; we include a
direct proof.

\begin{proposition}
  \label{planar:2d}
  For every planar graph~$G$, there exists a flat convex-polyhedral
  surface \Sur such that $\Gra(\Sur) \simeq G$.  Moreover, such a
  surface can be computed in linear time.
\end{proposition}

\begin{proof}
  Let $G$ be a planar embedded graph with at least three vertices (for
  at most two vertices the statement is trivially true).  We use a
  linear time algorithm by Read~\cite{r-nmdpg-87} to find a
  biconnected augmentation of $G$ on the same vertex set.  For each
  face of the resulting graph, we now add a new vertex and connect it
  to all vertices of the face by adding further edges.  This can be
  accomplished in linear time.  The resulting graph~$G'$ is a
  triangulation.  Let~$r$ be one of the vertices of $G'\setminus G$.

  The dual~$G^\star$ of~$G'$ is a cubic 3-connected planar graph.
  Using a linear-time algorithm by B{\'a}r{\'a}ny and
  Rote~\cite{br-scdpg-DM06}, we compute a planar drawing of~$G^\star$
  in which the boundary of each face is described by a strictly convex
  polygon and where the outer face is the face dual to~$r$.  Hence,
  the drawing is a flat convex-polyhedral surface~\Sur with
  $\Gra(\Sur)\simeq G'-r$.  By \cref{prop:subgraphsSubdivisions},
  there is also a representation of~$G$.  To compute it efficiently,
  observe that due to the $3$-regularity of $G^*$, the side trim
  operation defined in the proof of \cref{prop:subgraphsSubdivisions}
  can be carried out in constant time.  Moreover, the corresponding
  graph operation preserves the $3$-regularity.  Hence, it is easy to
  remove all unwanted adjancencies in linear time.  To obtain the
  desired representation of $G$, it remains to remove the polygons
  corresponding to the vertices of $G'\setminus G$, which can also be
  done in linear time.
\end{proof}

So for planar graphs, corner and side contacts behave similarly.  For
nonplanar graphs (for which the third dimension is essential), the
situation is different.  Here, side contacts are more restrictive.

\subsection{Complete Graphs}

We introduce the following notation.  In a polyhedral surface \Sur
with adjacency graph $G$, we denote by $P_v$ the polygon in~\Sur that
represents vertex~$v$ of~$G$.

\begin{lemma}
  \label{lemma:triangle}
  Let \Sur be a convex-polyhedral surface in~$\mathbb{R}^3$ with
  adjacency graph~$G$.  If $G$ contains a triangle $uvw$,
  polygons~$P_v$ and~$P_w$ lie in the same closed halfspace with
  respect to~$P_u$.
\end{lemma}

\begin{proof}
  Due to their convexity, each of $P_v$ and~$P_w$ lie entirely in one
  of the closed halfspaces with respect to the supporting plane
  of~$P_u$.  Moreover, one of the halfspaces contains both~$P_v$
  and~$P_w$; otherwise they cannot share a side and the edge $vw$
  would not be represented. (Recall that each side can be shared by at
  most two polygons. Thus, the side corresponding to the edge $uv$
  cannot simultaneously represent an adjacency with $w$.)
\end{proof}

A graph $H$ is \emph{subisomorphic} to a graph~$G$ if $G$ contains a
subgraph $G'$ with $H \simeq G'$.
Thomassen~\cite[page~98]{thomassenColorCriticalGraphs} has observed
the following.

\begin{proposition}[\cite{thomassenColorCriticalGraphs}]
  \label{obs:K5}
  There exists no convex-polyhedral surface~\Sur in~$\mathbb{R}^3$
  such that $K_5$ is subisomorphic to $\Gra(\Sur)$.
\end{proposition}

For completeness, we now prove Thomassen's observation.

\begin{proof}
  Suppose that there is a convex-polyhedral surface \Sur with
  $\Gra(\Sur) \simeq K_5$.  By \Cref{lemma:triangle} and the fact that
  all vertex triples form a triangle, the surface $\Sur$ lies in one
  closed halfspace of the supporting plane of every polygon~$P$ of
  \Sur.  In other words, \Sur is a subcomplex of a (weakly) convex
  polyhedron, whose adjacency graph must be planar. This yields a
  contradiction to the nonplanarity of $K_5$.  Together with
  \cref{prop:subgraphsSubdivisions} this implies the claim.
\end{proof}

Evans et al.~\cite{erssw-rghtp-GD19} showed that every bipartite graph
has a contact representation by touching polygons on a polynomial-size
integer grid in $\mathbb{R}^3$ for the case of corner contacts.  As we
have seen before, side contacts are less flexible.  In particular, in
\Cref{thm:k581} we show that $K_{5,81}$ cannot be represented.  On the
positive side, we show in the following that every (bipartite) graph
that comes from subdividing each edge of an arbitrary graph (at least)
once can be realized. In our construction, we place the polygons in a
cylindrical fashion, which is reminiscent of the realizations created
by Evans et al.  However, due to the more restrictive nature of side
contacts, the details of the two approaches are necessarily quite
different.

\begin{theorem}\label{obs:Knsub}
  Let $G$ be any graph, and let $G'$ be the subdivision of~$G$ in
  which every edge is subdivided with at least one vertex.  Then there
  exists a convex-polyhedral surface \Sur in~$\mathbb{R}^3$ such that
  $\Gra(\Sur) \simeq G'$.
\end{theorem}

\begin{proof}
  Let $V(G)=\{v_1,\dots,v_n\}$, let $E(G)=\{e_1,\dots,e_m\}$, and let
  $P$ be a strictly convex polygon with corners $p_1,\dots,p_{2m}$ in
  the plane.  We assume that $m \ge 2$, that $p_1$ and $p_{2m}$ lie on
  the x-axis, and that the rest of the polygon is a convex chain that
  projects vertically onto the line segment~$\overline{p_1p_{2m}}$,
  which we call the \emph{long side} of~$P$.  We call the other sides
  \emph{short sides}.  We choose~$P$ such that no short side is
  parallel to the long side.

  Let $Z$ be a (say, unit-radius) cylinder centered at the z-axis.
  For each vertex $v_i$ of~$G$, we take a copy~$P^i$ of~$P$ and place
  it vertically in~$\mathbb R^3$ such that its long side lies on the
  boundary of~$Z$; see \Cref{fig:1a}.  Each polygon~$P^i$ lies
  inside~$Z$ on a distinct halfplane that is bounded by the z-axis.
  Finally, all polygons are positioned at the same height, implying
  that for any $j \in \{1,\dots,2m\}$, all copies of~$p_j$ lie on the
  same horizontal plane~$h_j$ and have the same distance to the
  z-axis.

  \begin{figure}[htb]
    \begin{subfigure}[b]{0.48\textwidth}
      \centering \includegraphics[page=2]{Knsub2}
      \caption{polygons $P^1,\ldots,P^n$}
      \label{fig:1a}
    \end{subfigure}
    \hfill
    \begin{subfigure}[b]{0.48\textwidth}
      \centering \includegraphics{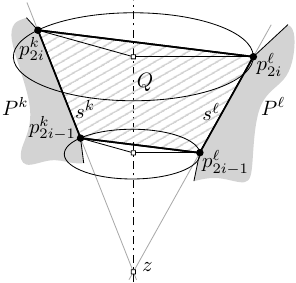}
      \caption{quadrilateral $Q$ spanned by $s^k$ and~$s^\ell$}
      \label{fig:1b}
    \end{subfigure}
    \caption{Illustration for the proof of \Cref{obs:Knsub}.}
    \label{fig:knsub}
  \end{figure}

  Let $i \in \{1,\dots,m\}$.  Then the side $s=p_{2i-1}p_{2i}$ is a
  short side of~$P$.  For $k=1,2,\dots,n$, we denote by~$s^k$ and
  $p_i^k$ the copies of~$s$ and $p_i$ in~$P^k$, respectively.  We
  claim that, for $1 \le k<\ell \le n$, the sides $s^k$ and $s^\ell$
  span a convex quadrilateral that does not intersect any $P^j$ with
  $j \not\in \{k,\ell\}$.  To prove the claim, we argue as follows;
  see \Cref{fig:1b}.
	
  By the placement of $P^k$ and $P^\ell$ inside $Z$, the supporting
  lines of~$s^k$ and~$s^\ell$ intersect at a point $z$ on the z-axis,
  implying that $s^k$ and $s^\ell$ are coplanar.  Moreover,
  $p_{2i-1}^k$ and $p_{2i-1}^\ell$ are at the same distance from $z$,
  and the same holds for $p_{2i}^k$ and $p_{2i}^\ell$.  Hence the
  triangle spanned by $z$, $p_{2i-1}^k$, and $p_{2i-1}^\ell$ is
  similar to the triangle spanned by $z$, $p_{2i}^k$, and
  $p_{2i}^\ell$, implying that $p_{2i-1}^kp_{2i-1}^\ell$ and
  $p_{2i}^kp_{2i}^\ell$ are parallel and hence span a convex
  quadrilateral~$Q$ (actually a trapezoid).  Finally, no polygon $P^j$
  with $j \not\in \{k,\ell\}$ can intersect $Q$ as any point in the
  interior of $Q$ lies closer to the z-axis than any point of $P^j$ at
  the same z-coordinate, which proves the claim.

  We use $Q$ as the polygon for the subdivision vertex of the
  edge~$e_i$ of $G$ (in case~$e_i$ was subdivided multiple times, we
  dissect~$Q$ accordingly).  Let $v_a$ and $v_b$ be the endpoints
  of~$e_i$.  By our claim, $Q$ does not intersect any $P^j$ with
  $j \not\in \{a,b\}$.  The quadrilateral $Q$ lies in the region
  of~$Z$ that is bounded by the horizontal planes~$h_{2i-1}$ and
  $h_{2i}$.  Since any two such regions are vertically separated and
  hence disjoint, the $m$ quadrilaterals together with the $n$ copies
  of $P$ constitute a valid representation of~$G'$.
\end{proof}

The combination of \Cref{obs:K5} and \Cref{obs:Knsub} rules out any
Kuratowski-type characterization for adjacency graphs of
convex-polyhedral surfaces.  This graph class contains a subdivision
of~$K_5$, but it does not contain~$K_5$; hence it is not minor-closed.
We remark that the subdivided $K_n$ has $\binom{n}{2}+n$ vertices and
crossing number $\Theta(n^4)$, so it is an adjacency graph of a
convex-polyhedral surface whose crossing number is quadratic in its
number of vertices (and edges).

\subsection{Complete Bipartite Graphs}

\begin{proposition}\label{prop:K44}
  There exists a convex-polyhedral surface~\Sur such that
  $\Gra(\Sur) \simeq K_{4,4}$.
\end{proposition}

\begin{proof}
  We describe how to obtain such a surface~\Sur.  We start with a
  rectangular box in $\mathbb{R}^3$ and stab it with two rectangles
  that intersect each other in the center of the box as indicated in
  \Cref{fig:K44} (a).

  \begin{figure}[tbh]
    \begin{subfigure}[t]{0.48\textwidth}
      \centering \includegraphics[height=6cm]{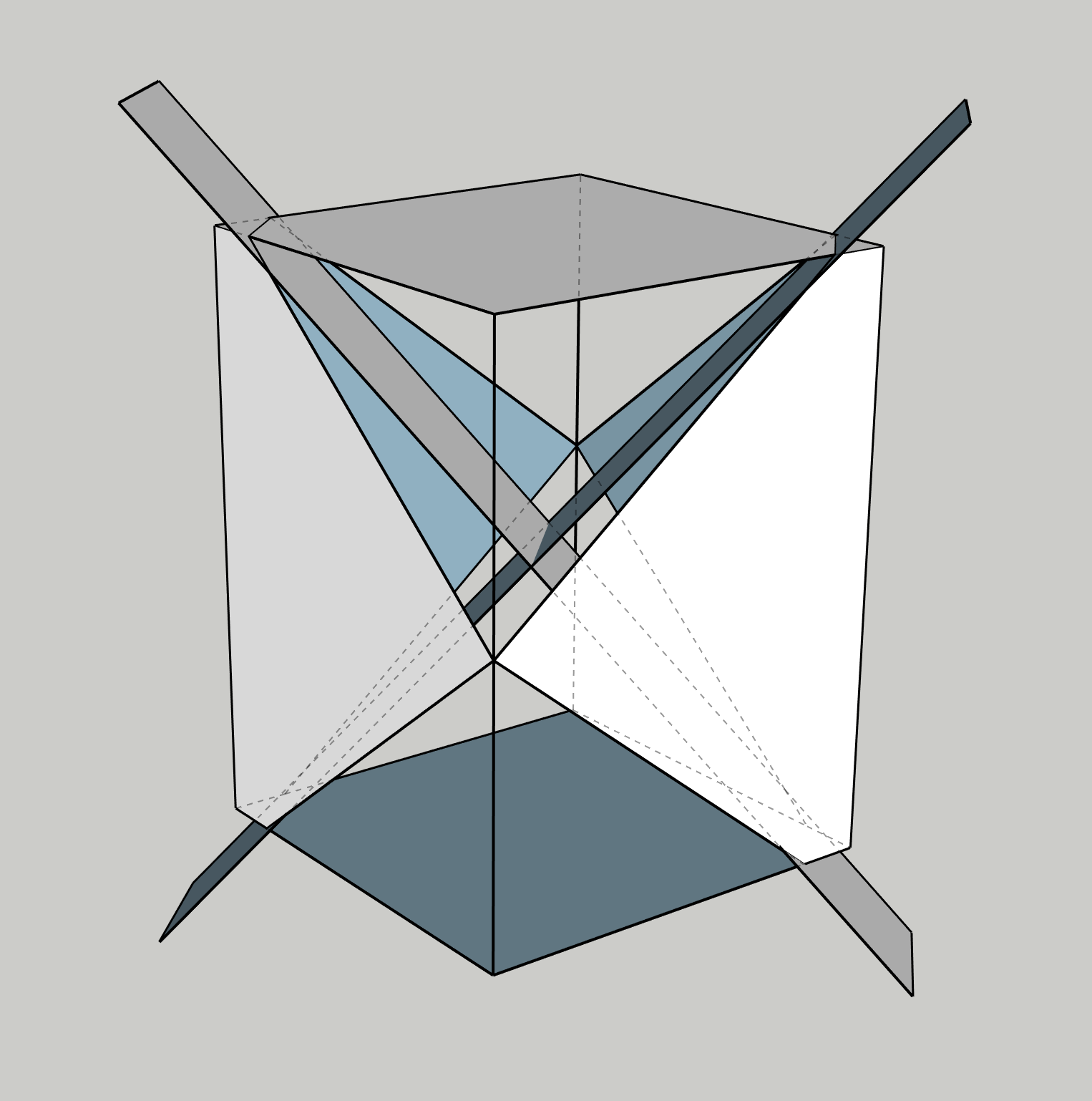}
      \caption{}
      \label{fig:K44A}
    \end{subfigure}
    \hfill
    \begin{subfigure}[t]{0.48\textwidth}
      \centering \includegraphics[height=6cm]{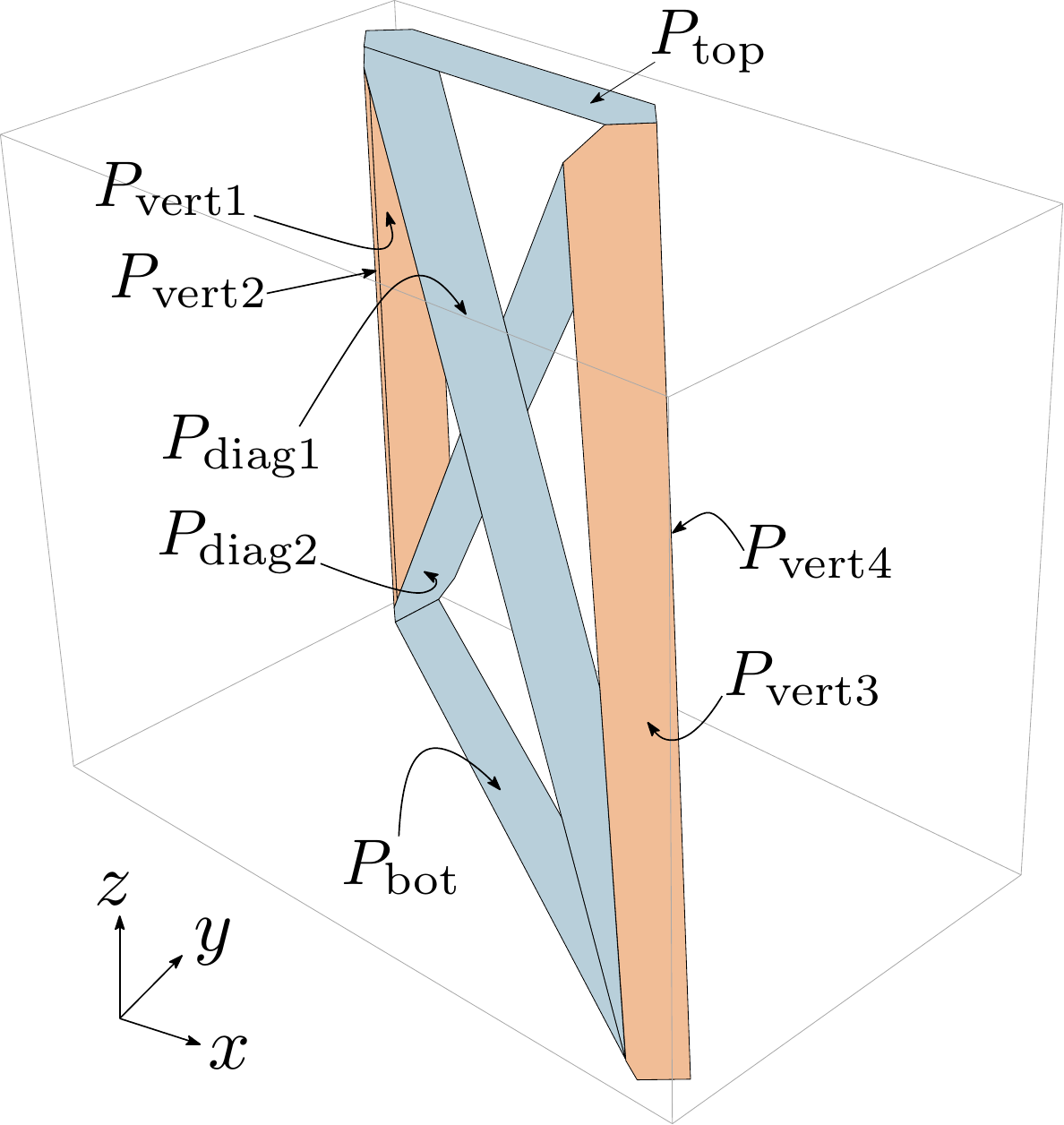}
      \caption{}
      \label{fig:K44C}
    \end{subfigure}
	
    \caption{Construction of a convex-polyhedral surface~\Sur with
      $\Gra(\Sur) \simeq K_{4,4}$.  Subfigure (a) depicts the
      rectangular box and the two slanted rectangles that form the
      basis of the construction. The figure illustrates the situation
      before the intersection of the two slanted rectangles is removed
      by shifting a corner.  Subfigure (b) illustrates the final
      realization.  The 2-coloring of the polygons reflects the
      bipartition of~$K_{4,4}$.  Note that in the depicted projection
      the polygons $P_{\mathrm{vert2}}$ and $P_{\mathrm{vert4}}$ are
      shown as line-segments / very thin polygons; for a better view
      of these polygons, refer to \cref{fig:k44additional}.  The
      orientation of the coordinate system is indicated in the
      bottom-left of the figure.}
    \label{fig:K44}
  \end{figure}

  We can now draw polygons on these eight rectangles such that each of
  the four vertical rectangles (representing the four vertices of one
  class of the bipartition of~$K_{4,4}$) contains a polygon that has a
  side contact with a polygon on each of the four horizontal or
  slanted rectangles (representing the other class of the bipartition
  of $K_{4,4}$).  To remove the intersection of the (polygons drawn on
  the) two slanted rectangles, we shift one corner of the original
  box; see \Cref{fig:K44}~(b) and \cref{fig:k44additional}.  We refer
  to the two horizontal, the four vertical, and the two slanted
  polygons as $P_\text{top}$, $P_\text{bot}$, $P_\text{vert1}$,
  $P_\text{vert2}$, $P_\text{vert3}$, $P_\text{vert4}$,
  $P_\text{diag1}$, and $P_\text{diag2}$, respectively, and list the
  coordinates of their corners in Table~\ref{tab:coork44}.

  \begin{table}[ht]
    \centering
    \caption{Coordinates for the polyhedral complex realizing
      $K_{4,4}$ as adjacency graph.}
    \label{tab:coork44}

  \begin{tabular}{lp{11cm}}
    \toprule
    polygon & vertices \\
    \midrule
    $P_\text{top}$ & $(-6,0,10),(-5,-1,10),(5,-1,10),(6,0,10),$ $(5,1,10),(-5,1,10)$\\
    $P_\text{bot}$ & $(-5,1,-10),(-6,0,-10),(-5,-1,-10),(17,-13,-10),$ $(18,-12,-10),(17,-11,-10)$\\
    $P_\text{vert1}$ & $(-6,0,10),(-5,1,10),\left(-\frac{23}{5},\frac{7}{5},\frac{106}{11}\right),\left(-\frac{23}{5},\frac{7}{5},-\frac{46}{5}\right),$ $(-5,1,-10),(-6,0,-10)$\\
    $P_\text{vert2}$ &	$(-6,0,10),(-5,-1,10),\left(-\frac{23}{5},-\frac{7}{5},\frac{106}{11}\right),\left(-\frac{23}{5},-\frac{7}{5},-\frac{46}{5}\right),$ $(-5,-1,-10),(-6,0,-10)$\\
    $P_\text{vert3}$ &	$(6,0,10),(5,-1,10),\left(\frac{23}{5},-\frac{59}{25},\frac{46}{5}\right),\left(\frac{82}{5},-\frac{712}{55},-\frac{104}{11}\right),$ $(17,-13,-10),(18,-12,-10)$\\
    $P_\text{vert4}$ & $(6,0,10),(5,1,10),\left(\frac{23}{5},\frac{7}{5},\frac{46}{5}\right),\left(\frac{82}{5},-\frac{52}{5},-\frac{104}{11}\right),$ $(17,-11,-10),(18,-12,-10)$\\
    $P_\text{diag1}$& $(-5,1,10),\left(-\frac{23}{5},\frac{7}{5},\frac{106}{11}\right),\left(\frac{82}{5},-\frac{52}{5},-\frac{104}{11}\right),(17,-11,-10),$ $(17,-13,-10),\left(\frac{82}{5},-\frac{712}{55},-\frac{104}{11}\right),\left(-\frac{23}{5},-\frac{7}{5},\frac{106}{11}\right),(-5,-1,10)$\\
    $P_\text{diag2}$&
                      $(-5,1,-10),\left(-\frac{23}{5},\frac{7}{5},-\frac{46}{5}\right),\left(\frac{23}{5},\frac{7}{5},\frac{46}{5}\right),(5,1,10),(5,-1,10),$ $\left(\frac{23}{5},-\frac{59}{25},\frac{46}{5}\right),\left(-\frac{23}{5},-\frac{7}{5},-\frac{46}{5}\right),(-5,-1,-10)$ \\
    \bottomrule
  \end{tabular}
\end{table}
 
With the specified coordinates, $P_\text{diag1}$ and $P_\text{diag2}$
each have a side that lies in the interior of~$P_\text{top}$ and a
side that lies in the interior of~$P_\text{bot}$.  To fix this, one
needs to clip the two polygons such that they lie in the interior of
the original cuboid.  This can be done by intersecting them with the
slab $-9.9 \le z \le 9.9$.  Moreover, the two polygons
$P_\text{vert1}$ and $P_\text{vert2}$ (as well as $P_\text{vert3}$ and
$P_\text{vert4}$) have a common side, even though they correspond to
vertices in the same class of the bipartition of $K_{4,4}$.  These
unwanted contacts can also be removed by slightly clipping the
polygons, cf.~\cref{prop:subgraphsSubdivisions}.
\end{proof}

\begin{figure}[htb]
  \begin{subfigure}[t]{0.38\textwidth}
    \centering \includegraphics[scale=.28]{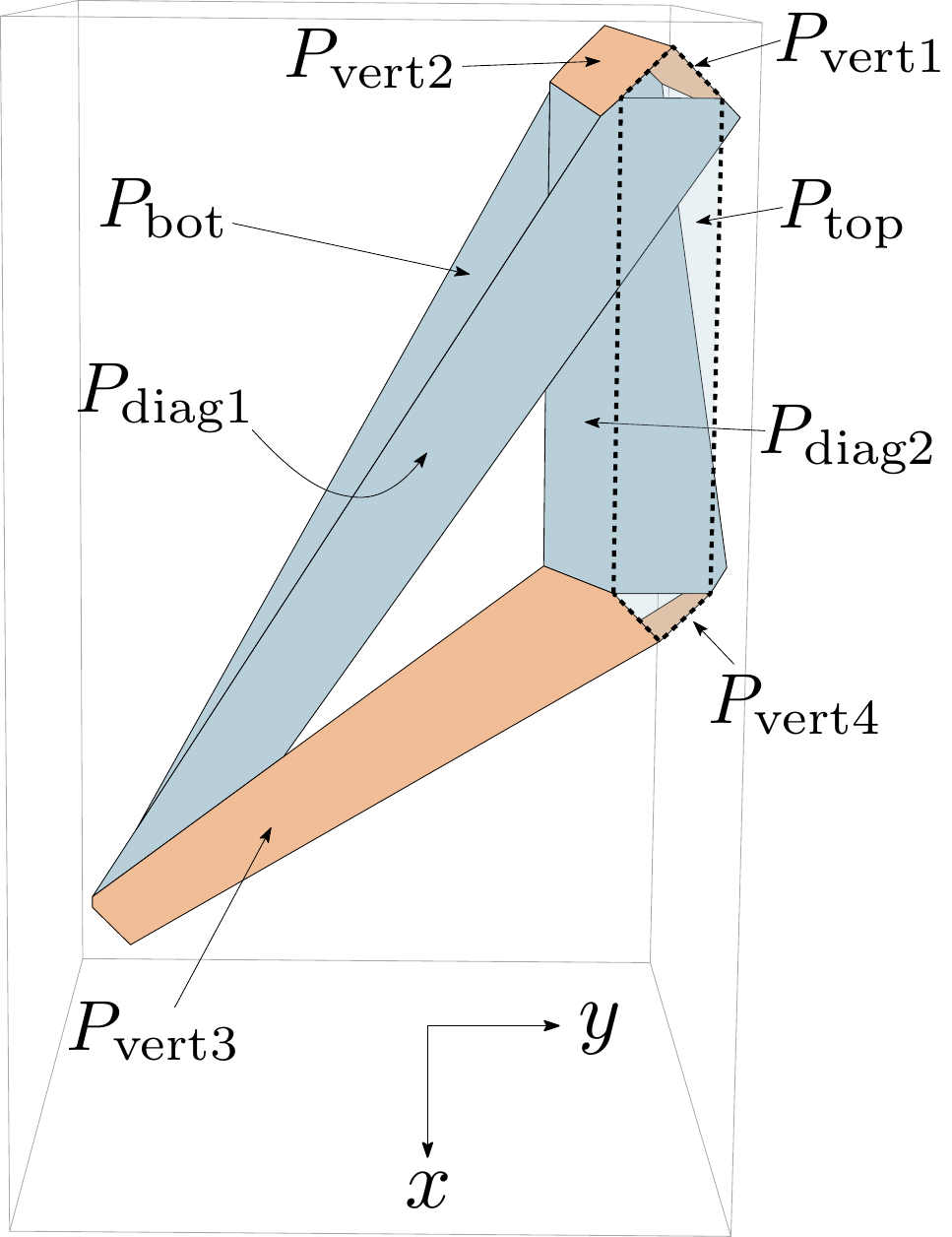}
    \caption{view from above; the polygon~$P_\text{top}$ is
      transparent and dashed}
    \label{fig:topwithoutlid}
  \end{subfigure}
  \hfill
  \begin{subfigure}[t]{0.41\textwidth}
    \centering \includegraphics[scale=.4]{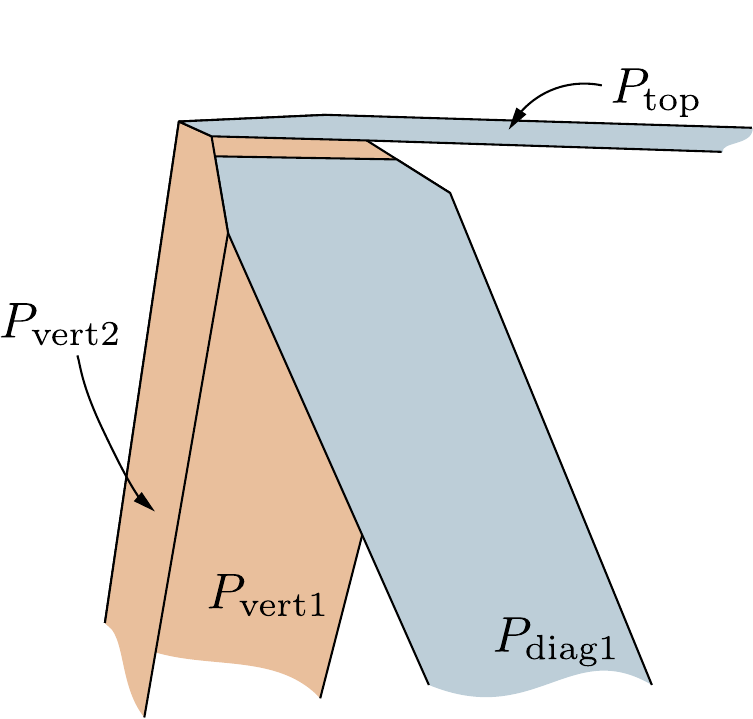}
    \caption{close-up after clipping; the separation
      between~$P_\text{top}$ and~$P_\text{diag1}$ is exaggerated}
    \label{fig:withoutdiags}
  \end{subfigure}

  \caption{Additional views of the realization of $K_{4,4}$.}
  \label{fig:k44additional}
\end{figure}

\begin{proposition}
  \label{prop:K35}
  There exists a convex-polyhedral surface~\Sur such that
  $\Gra(\Sur) \simeq K_{3,5}$.
\end{proposition}

\begin{proof}
  We call the vertices of the smaller bipartition class the gray
  vertices, and their polygons gray polygons. For the other class we
  pick a distinct color for every vertex and use the same
  naming-by-color convention. We start our construction with a
  triangular prism in which the quadrilateral faces $q_1,q_2,q_3$ are
  rectangles of the same size. Each of the faces $q_i$ will contain
  one gray polygon. All colorful polygons lie inside the prism.  We
  call the lines resulting from the intersection of the supporting
  planes with the prism the \emph{colorful supporting lines}.
  Unfolding the faces~$q_1$, $q_2$, and~$q_3$ into the plane yields
  \Cref{fig:k35unfolded}, which shows the gray polygons and the
  colorful supporting lines. Note that the vertices of the gray
  polygons in the figure are actually very small edges that have the
  slope of the colorful supporting line on which they are placed.  The
  colorful polygons are now already determined.

  \begin{figure}[p]
    \begin{minipage}{\textwidth}
      \centering %
      \includegraphics[page=2]{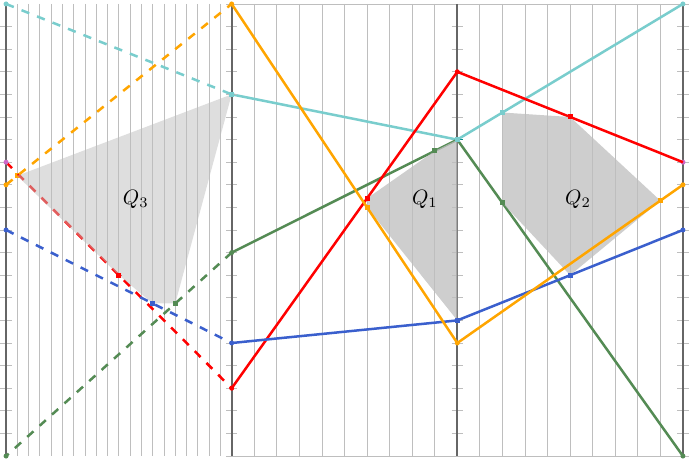}
      \caption{Constructing a convex-polyhedral surface whose
        adjacency graph is isomorphic to~$K_{3,5}$.  The prism
        is unfolded into the plane.  All polygon vertices are
        contained in the vertical grid lines.}
      \label{fig:k35unfolded}
    \end{minipage}

    \bigskip

    \begin{minipage}{\textwidth}
      \centering %
      \includegraphics[page=6]{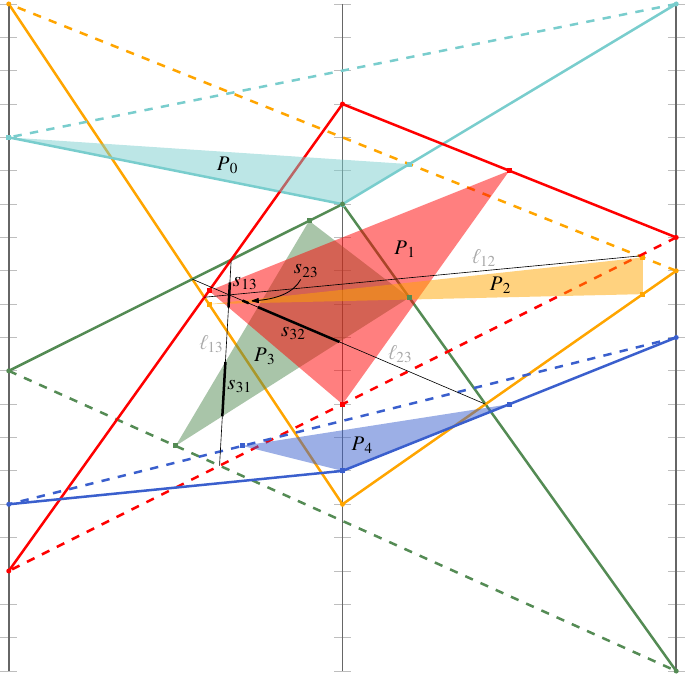}
      \caption{Front view of the prism containing a realization
        of~$K_{3,5}$.  Lines on the back face are dashed.  The thin
        black lines are the lines of intersection among the supporting
        planes of the red, yellow, and green polygons.  The thick
        black line segments indicate which parts of the intersection
        lines are contained in a colorful polygon.  Since the segments
        are disjoint, the colorful polygons are disjoint, too.}
      \label{fig:k35front}
    \end{minipage}
  \end{figure}
  
  It remains to check that the colorful polygons are disjoint.
  \Cref{fig:k35front} shows the prism in a view from the side where we
  dashed all objects on the hidden prism face.  The cyan polygon~$P_0$
  and the blue polygon~$P_4$ avoid all other colorful polygons in this
  projection and thus they avoid all other polygons in~$\mathbb{R}^3$,
  too.

  For the red polygon~$P_1$, the orange polygon~$P_2$ and the green
  polygon~$P_3$, we proceed as follows to prove disjointness.  Pick
  two of the polygons and name them~$P_i$ and~$P_j$.  The
  line~$\ell_{ij}$ of intersection of the supporting planes of~$P_i$
  and~$P_j$ is determined by the two intersections of the
  corresponding colorful supporting lines.  If the polygons intersect,
  they have to intersect on this line.  Polygon~$P_i$
  intersects~$\ell_{ij}$ in a segment $s_{ij}$; polygon~$P_j$
  intersects~$\ell_{ij}$ in~$s_{ji}$.  \cref{fig:k35front} shows,
  however, that~$s_{ij}$ and~$s_{ji}$ do not overlap in any of the
  three cases.  (Note that~$P_2$ does not intersect~$\ell_{12}$ and
  hence $P_2$ does not intersect~$P_1$ either.)  We remark that our
  construction can be verified easily since it is grid-based in the
  following sense.  First, note that each intersection point of the
  supporting plane of a colorful polygon and one of the three vertical
  edges of the prism has integral height in $\{0,\dots,20\}$; see the
  tics in \Cref{fig:k35unfolded}.  Second, on each of the three
  vertical faces of the prism, we define a set of equally-spaced
  vertical lines (10 on the two front faces, 20 on the back face) such
  that each polygon vertex lies on the intersection of one of these
  lines and its supporting plane.
\end{proof}

In contrast to \Cref{prop:K44,prop:K35}, we can show that not every
complete bipartite graph can be realized as a convex-polyhedral
surface in~$\mathbb{R}^3$.

\begin{theorem}\label{thm:k581}
  There exists no convex-polyhedral surface~\Sur in~$\mathbb{R}^3$
  such that $ K_{5,81}$ is subisomorphic to $\Gra(\Sur)$.
\end {theorem}

To prove the theorem we start with some observations about realizing
complete bipartite graphs. We will consider a set $R$ of red polygons,
and a set $B$ of blue polygons, so that each red--blue pair must have
a side contact.  For each $p\in R\cup B$, we denote by $p^=$ the
supporting plane of $p$, by $p^-$ the closed half-space left of $p^=$,
and by $p^+$ the closed half-space right of $p^=$ (orientations can be
chosen arbitrarily).  We start with a simpler setting where we have an
additional constraint.  We call $B$ \emph{one-sided with respect to
  $R$} if, for each blue polygon~$b$, all red polygons lie in the same
half-space with respect to~$b$, i.e.,
$\forall b \in B \colon ((\forall r \in R \colon r \subseteq b^-) \vee
(\forall r \in R \colon r \subseteq b^+))$.

\begin{lemma}
  \label{lem:oneside}
  Let $R$ and $B$ be two sets of convex polygons in $\mathbb{R}^3$
  realizing $K_{|R|,|B|}$. If $|R|=3$ and $B$ is one-sided with
  respect to $R$, then $|B|\leq 8$.
\end{lemma}

\begin{proof}
  Let $R = \{r_1, r_2, r_3\}$ and let $\mathcal A$ be the arrangement
  of the supporting planes of $R$.  Assume that $B$ is one-sided with
  respect to $R$ and consider a polygon $b \in B$. For every polygon
  $r_i \in R$, since $b$ is convex and shares a side with $r_i$, $b$
  is contained in $r_i^-$ or $r_i^+$.
  Thus, $b$ is contained in a (closed) cell $C$ of $\mathcal A$.  Let
  $r_*=r_1^=\cap r_2^=\cap r_3^=$ be the intersection of the
  supporting planes of $R$.
   
  We will first argue about the case where $r_*$ is not a point.  We
  may assume that no two supporting planes of $R$ coincide; otherwise,
  by strict convexity, two coplanar red polygons imply that all blue
  polygons lie in the same plane. Moreover, if $|B|\ge 2$, it follows
  symmetrically that all red polygons are coplanar. Hence, all
  polygons must lie in the same plane and the non-planarity of
  $K_{3,3}$ implies that $|B| \le 2$.  It follows that $r_*$ is not a
  plane.  Further, if $r_*$ is a line, then $C$ has only two bounding
  planes and therefore one of the red polygons is only present as a
  subset of $r_*$; see~\Cref{fig:lem10-casesA}. This implies that $b$
  has a side on $r_*$ and on each of the open half-planes bounding
  $C$, which is impossible.  Finally, if $r_*=\emptyset$, we can apply
  a projective transformation such that the bounding planes of the
  three red polygons intersect. Therefore, we can assume that
  $r_*\neq \emptyset$.
  
  It remains to consider the case where $r_*$ is a point, in which
  case the arrangement~$\mathcal{A}$ defines eight (closed) cells,
  called \emph{octants}, of the form
  $Q^{\alpha_0 \beta_0 \gamma_0} = r_1^{\alpha_0} \cap r_2^{\beta_0}
  \cap r_3^{\gamma_0}$, where
  $\alpha_0,\beta_0,\gamma_0\in\{\mathtt{+},\mathtt{-}\}$.  We
  distinguish two subcases: either (1)~no polygon in~$R$ contains the
  point~$r_*$ (see \cref{fig:lem10-casesB}) or there is a red polygon
  whose supporting plane contains a blue polygon, or (2)~$r_*$ is
  contained in a (single) polygon in~$R$ (see \cref{fig:lem10-casesC})
  and there is no red polygon whose supporting plane contains a blue
  polygon.
  
  \begin{figure}[htb]
    \begin{subfigure}[t]{0.27\textwidth}
      \centering \includegraphics[page=5]{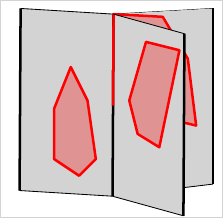}
      \caption{$r_*$ is a line}
      \label{fig:lem10-casesA}
    \end{subfigure}
    \hfill
    \begin{subfigure}[t]{0.33\textwidth}
      \centering \includegraphics[page=2]{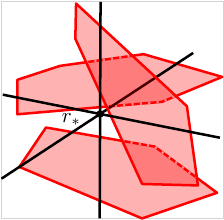}
      \caption{no red polygon contains $r_*$ (Case~1)}
      \label{fig:lem10-casesB}
    \end{subfigure}
    \hfill
    \begin{subfigure}[t]{0.33\textwidth}
      \centering \includegraphics[page=3]{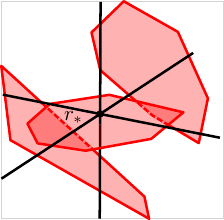}
      \caption{a red polygon contains $r_*$ (Case~2)}
      \label{fig:lem10-casesC}
    \end{subfigure}
    \caption{The three red polygons and the intersection $r_*$ of
      their supporting planes.}
    \label{fig:lem10-cases}
  \end{figure}

  {\vspace{.5em}\noindent\bf Case~1:} \emph{No polygon in~$R$ contains
    the point $r_*$ or there is a red polygon whose supporting plane
    contains a blue polygon.}  Our plan is to show that there are four
  octants whose union contains all blue polygons and that each of
  these four octants contains at most two blue polygons, which implies
  that the total number of blue polygons is bounded by $4\cdot 2=8$,
  as claimed.

  We start to show that there are four octants whose union contains
  all blue polygons.  To this end, we distinguish two subcases.
  
  {\vspace{.5em}\noindent\bf Case~1.1:} \emph{No polygon in~$R$
    contains the point $r_*$.}  Since the point $r_*$ is disjoint from
  all red polygons, each red polygon lies on the boundary of at most
  six octants.  More precisely, there are signs
  $\alpha_2, \alpha_3, \beta_1, \beta_3, \gamma_1, \gamma_2 \in
  \{\mathtt{+},\mathtt{-}\}$ such that $r_1$ cannot intersect the two
  octants $Q^{\pm \beta_1 \gamma_1}$, $r_2$ cannot intersect the two
  octants $Q^{\alpha_2{\pm}\gamma_2}$, and $r_3$ cannot intersect the
  two octants $Q^{\alpha_3 \beta_3 \pm}$.  It is now easy to verify
  that there are at most four octants that intersect all three red
  polygons.  For example, if
  $\alpha_2=\alpha_3=\beta_1=\beta_3=\gamma_1=\gamma_2= {\mathtt{+}}$,
  then only the octants $Q^{+--},Q^{-+-},Q^{--+},$ and $Q^{---}$ can
  intersect all three red polygons.  Since every blue polygon has to
  be contained in one of these four octants, the claim follows.
 
  {\vspace{.5em}\noindent\bf Case~1.2:} \emph{There is a red polygon,
    say~$r_1$, whose supporting plane~$r_1^=$ contains a blue polygon,
    say~$b_1$.}  Recall that we assume that no two supporting planes
  of red polygons coincide.  Symmetrically, we may assume that no two
  supporting planes of blue polygons coincide.  Hence, without loss of
  generality, we may assume that~$r_2$ and~$r_3$ intersect the
  interior of $b_1^+$, which, without loss of generality, coincides
  with $r_1^+$.  Consequently, all blue polygons are contained in
  $r_1^+$ and, thus, they are contained in the union of the four
  corresponding octants.
   
  \vspace{.5em} So far, we have shown that there are four octants
  whose union contains all blue polygons.  Now consider one octant~$C$
  that contains a blue polygon (and is thus incident to all~$r_i$).
  For the argument within this cell, we can truncate all $r_i$ to $C$.
  We claim that there can be at most two blue polygons in $C$. Assume
  towards a contradiction that we have three such polygons~$b_1$,
  $b_2$, and~$b_3$.  For each $i \in \{1,2,3\}$, the polygon~$b_i$ has
  three sides in common with $r_1,r_2,r_3$.  Let~$b_i'$ be the convex
  hull of these three sides.  Consider now the set
  $P = \{r_1,r_2,r_3,b_1',b_2',b_3'\}$.  Let $b'_i$ and $b'_j$ be two
  different (partial) blue polygons.  Since $b^=_i$ has all (sides of)
  red polygons on one side, it has also the three sides defining
  $b'_j$ on one side. Hence, $b_j'$ is on one side of the supporting
  plane of $b'_i$, and this side is the same for all polygons in
  $P\setminus\{b'_i\}$. On the other hand, by the definition of~$C$,
  every $r^=_i$ has all polygons in $P\setminus\{r_i\}$ on one common
  side.  Thus, the polygons in~$P$ are in convex position.  Consider
  now the convex hull $\mathcal{H}$ of~$P$.  We get that $\mathcal{H}$
  is a convex polyhedron with the polygons of $P$ embedded on its
  surface.  We can draw the contact graph of $P$ on that surface
  without crossings.  Since the surface is homeomorphic to a sphere,
  we obtain a contradiction since the contact graph is a $K_{3,3}$ and
  therefore nonplanar.  Thus, any octant can contain at most two blue
  polygons, as claimed.
  
  Altogether, we have shown that there are four octants whose union
  contains all blue polygons and that each of these four octants
  contains at most two blue polygons, which implies that the total
  number of blue polygons is bounded by $4\cdot 2=8$ in Case~1.
  
  {\vspace{.5em}\noindent\bf Case~2:} \emph{A red polygon
    contains~$r_*$ and there is no red polygon whose supporting plane
    contains a blue polygon.}  Without loss of generality,
  $r_* \in r_1$.  Similar to Case~1.1, the polygons~$r_2$ and~$r_3$
  both lie on the boundary of at most six octants, which implies that
  there are at most five octants that intersect every polygon in~$R$.
  We claim that at most one polygon of~$B$ can intersect any given
  octant.  Consider an octant~$Q$ and assume that it is intersected by
  two one-sided blue polygons $b_1$ and $b_2$.  Note that $Q$ is
  bounded by three (unbounded) faces~$f_1$, $f_2$, and~$f_3$ such
  that, for $j \in \{1,2,3\}$, the red polygon~$r_j$ (truncated
  to~$Q$) lies in~$f_j$.
  
  Let $\ell_{i,j}$ denote the intersection of the plane~$b^=_i$ and
  the face~$f_j$. Because there is no red polygon whose supporting
  plane contains a blue polygon, $\ell_{i,j}$ is not the entire
  face~$f_j$ but a segment or a ray. Let $t_i$ be the trace formed by
  $\ell_{i,1}$, $\ell_{i,2}$, $\ell_{i,3}$.  Then~$t_i$ is either a
  triangle or the concatenation of two rays and a segment; see
  \Cref{fig:lem10-onecellC,fig:lem10-onecellE}.  Note that the
  traces~$t_1$ and~$t_2$ intersect in at most two points.
 
  Given two different faces $f_i$ and~$f_j$ of~$Q$, we call their
  intersection $f_i \cap f_j$ an \emph{axis} of~$Q$.  Note that each
  trace intersects at least two of the three axes of~$Q$.
  Hence, there is a face~$f_j$ such that~$\ell_{1,j}$ and~$\ell_{2,j}$
  have endpoints on the same axis contained in~$f_j$.  We complete the
  proof by distinguishing two subcases, depending on the intersection
  of $\ell_{1,j}$ and $\ell_{2,j}$.

  {\vspace{.5em}\noindent\bf Case~2.1:} \emph{$\ell_{1,j}$ and
    $\ell_{2,j}$ do not intersect in the relative interior of $f_j$;
    see \cref{fig:lem10-onecellF}.}  Then one of them (say
  $\ell_{1,j}$) separates the other (say $\ell_{2,j}$) from $r_*$ on
  $f_j$.  Because $r_j$ lies between~$\ell_{1,j}$ and~$\ell_{2,j}$, we
  get that $\ell_{1,j}$ separates~$r_*$ from~$r_j$.  This is a
  contradiction to the fact that $b_1$ is one-sided.  Thus, there are
  at most $1 \cdot 5 = 5$ polygons in $B$ in Case~2.1.

  {\vspace{.5em}\noindent\bf Case~2.2:} \emph{$\ell_{1,j}$ and
    $\ell_{2,j}$ intersect in the relative interior of $f_j$.}  Then
  there is a $j' \in \{1,2,3\} \setminus \{j\}$ such that
  $\ell_{1,j'}$ and $\ell_{2,j'}$ have endpoints on the same axis and
  do not intersect; otherwise the traces intersect three times.
  Hence, if we replace~$j$ by~$j'$, we are in Case~2.1.
\end {proof}
  
\begin{figure}[htb]
  \begin{subfigure}[t]{0.24\textwidth}
    \centering \includegraphics[page=7]{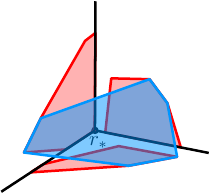}
    \caption{}
    \label{fig:lem10-onecellB}
  \end{subfigure}
  \hfil
  \begin{subfigure}[t]{0.24\textwidth}
    \centering \includegraphics[page=8]{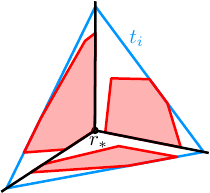}
    \caption{}
    \label{fig:lem10-onecellC}
  \end{subfigure}
  \hfil
  \begin{subfigure}[t]{0.24\textwidth}
    \centering \includegraphics[page=12]{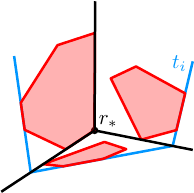}
    \caption{}
    \label{fig:lem10-onecellE}
  \end{subfigure}
  \hfil
  \begin{subfigure}[t]{0.24\textwidth}
    \centering \includegraphics[page=13]{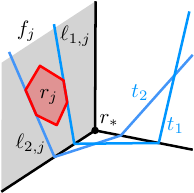}
    \caption{}
    \label{fig:lem10-onecellF}
  \end{subfigure}
  \caption{Illustration of Case~2 for the proof of \cref{lem:oneside}.
    (a) A single octant (towards the viewer) with three truncated red
    polygons and a possible blue polygon within the octant that has a
    side contact with all three red polygons, (b)~the trace of a
    possible blue polygon forming a triangle, (c) the trace of a
    possible blue polygon consisting of two rays and a segment, and
    (d)~two traces yield at least one blue polygon that is not
    one-sided.}
  \label{fig:lem10-onecell}
\end{figure}

With the help of \Cref{lem:oneside} we can now prove \Cref{thm:k581}.

\begin{proof}[Proof of \Cref{thm:k581}]
  Assume that $K_{5,81}$ can be realized, and let $R$ be a set of five
  red polygons.  Since every $b \in B$ is adjacent to all polygons in
  $R$, $b$ partitions $R$ into two sets: those in $b^-$ and those in
  $b^+$. At least one of these subsets must have at least three
  elements. Arbitrarily charge $b$ to such a set of three polygons.
  By \Cref{lem:oneside}, each set of three red polygons can be charged
  at most eight times. There are ${5 \choose 3} = 10$ sets of three
  red polygons.  Therefore, there can be at most $8 \cdot 10 = 80$
  blue polygons; a contradiction. Together with
  \Cref{prop:subgraphsSubdivisions} this implies the claim.
\end {proof}

\subsection{3-Trees}

The graph class of $3$-trees is recursively defined as follows: $K_4$
is a $3$-tree.  A graph obtained from a $3$-tree~$G$ by adding a new
vertex~$x$ with exactly three neighbors~$u,v,w$ that form a triangle
in~$G$ is a $3$-tree.  We say~$x$ is \emph{stacked} on the
triangle~$uvw$.  It follows that for each $3$-tree there exists a (not
necessarily unique) \emph{construction sequence} of $3$-trees
$G_4,G_5,\dots,G_n$ such that $G_4\simeq K_4$, $G_n=G$, and where for
$i=4,5,\dots, n-1$ the graph $G_{i+1}$ is obtained from~$G_i$ by
stacking a vertex~$v_{i+1}$ on some triangle of~$G_i$.

By \Cref{planar:2d}, for every planar $3$-tree~$G$ there is a
polyhedral surface~\Sur (even in~$\mathbb{R}^2$) with
$\Gra(\Sur) \simeq G$.  On the other hand, we can show that no
nonplanar $3$-tree has such a realization in $\mathbb{R}^3$.  To this
end, we observe that a $3$-tree is nonplanar if and only if it
contains the \emph{\TheThreeTree} as a subgraph.  The \TheThreeTree is
the graph that consists of $K_{3,3}$ plus a cycle that connects the
vertices of one part of the bipartition; see \Cref{fig:3tree}.  We
show that the \TheThreeTree is not realizable.

\begin{figure}[htb]
  \centering \includegraphics{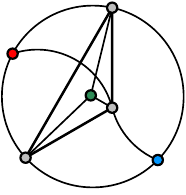}
  \caption{The unique minimal nonplanar 3-tree, which we call
    \TheThreeTree.}
  \label{fig:3tree}
\end{figure}

\begin{lemma}\label{prop:stackOnBothSides}
  Let~$uvw$ be a separating triangle in a plane $3$-tree $G=(V,E)$.
  Then there exist vertices~$a,b\in V$ that belong to distinct sides
  of~$uvw$ in~$G$ such that both $\{a,u,v,w\}$ and $\{b,u,v,w\}$
  induce a~$K_4$ in~$G$.
\end{lemma}

\begin{proof}
  Let $G_4,G_5,\dots,G_n$ denote a construction sequence of~$G=G_n$,
  and let~$k$ be the largest index in $\{4,5,\dots,n\}$ such that
  $uvw$ is nonseparating in~$G_k$.  Since~$uvw$ is separating
  in~$G_{k+1}$, it follows that the vertex~$v_{k+1}=a$ is stacked
  on~$uvw$ (say, inside $uvw$) to obtain~$G_{k+1}$ and, hence,
  $\{a,u,v,w\}$ induce a~$K_4$ in~$G_{k+1}$ and~$G$.
 
  It remains to argue about the existence of the vertex~$b$ in the
  exterior of~$uvw$.  If~$uvw$ is one of the triangles of the
  original~$G_4\simeq K_4$, there is nothing to show, so assume
  otherwise.  Let~$j$ be the smallest index in $\{5,6,\dots, n\}$ such
  that $uvw$ is contained in~$G_j$.  It follows that one of~$u,v,w$,
  say~$u$, is the vertex~$v_j$ that was stacked on some triangle~$xyz$
  of~$G_{j-1}$ to obtain~$G_j$.  Without loss of generality, we may
  assume that~$\{v,w\}=\{y,z\}$.  It follows that~$x=b$ forms a~$K_4$
  with~$u,v,w$ in~$G_j$ and~$G$.
\end{proof}

\begin{lemma}\label{lem:characterizePlanarThreeTree}
  A $3$-tree is nonplanar if and only if it contains the \TheThreeTree
  as a subgraph.
\end{lemma}

\begin{proof}
  The \TheThreeTree is nonplanar because it contains a~$K_{3,3}$ (one
  part of the bipartition is formed by the gray vertices and the other
  by the colored vertices).

  For the other direction, let~$G$ be a nonplanar $3$-tree.  Let
  $G_4,G_5,\dots,G_n$ be a construction sequence of~$G$.  Let~$k$ be
  the smallest index in $\{4,5,\dots, n\}$ such that~$G_k$ is
  nonplanar.  By $3$-connectivity, the graph~$G_{k-1}$, which is
  planar, has a unique combinatorial embedding.  Therefore, we may
  consider~$G_{k-1}$ to be a plane graph.  Let~$uvw$ be the triangle
  that the vertex~$v_k$ was stacked on to obtain~$G_k$ from~$G_{k-1}$.
  Since~$G_k$ is nonplanar, the triangle~$uvw$ is a separating
  triangle of~$G_{k-1}$.  It follows by \Cref{prop:stackOnBothSides}
  that~$G_k$ (and, hence,~$G$) contains the \TheThreeTree.
\end{proof}

\begin{lemma}\label{lem:nonrealizable_3treeB}
  There exists no convex-polyhedral surface~\Sur in~$\mathbb{R}^3$
  such that the \TheThreeTree is subisomorphic to $\Gra(\Sur)$.
\end{lemma}

\begin{proof}
  We refer to the vertices of the \TheThreeTree as the three gray
  vertices and the three colored (red, green, and blue) vertices; see
  also \Cref{fig:3tree}.  Given the correspondence between vertices
  and polygons (and their supporting planes), we also refer to the
  polygons (and the supporting planes) as gray and colored.
  
  Assume that the \TheThreeTree can be realized.  Consider the
  arrangement of the gray supporting planes.  By strict convexity, it
  follows that if a pair of gray polygons has the same supporting
  plane, then all their common neighbors lie in the same plane. This
  implies that all supporting planes coincide -- a contradiction to
  the non-planarity of the \TheThreeTree.  Consequently, the gray
  supporting planes are pairwise distinct. (Likewise, it holds that no
  colored and gray supporting plane coincide.)
  
  We now argue that all colored polygons are contained in the same
  closed cell of the gray arrangement.  To see this, fix one gray
  polygon and observe, by \Cref{lemma:triangle}, that all polygons are
  contained in the same closed half space with respect to its
  supporting plane.
	
  Note that the gray plane arrangement has one of the following two
  combinatorics: either the three planes have a common point of
  intersection (cone case) or not (prism case).  In the first case,
  the planes partition the space into eight cones, one of which
  contains all polygons; in the second case, the (unbounded) cell
  containing all polygons forms a (unbounded) prism.  For a unified
  presentation, we transform any occurrence of the first case into the
  second case. To do so, we move the apex of the cone containing all
  polygons to the plane at infinity by a projective
  transformation. This turns each face of the cone into a strip that
  is bounded by two of the extremal rays of the cone, which we now
  have deformed into a prism.
	
  Consider one of the strips, which we call~$S$.  The strip~$S$ has to
  contain one of the gray polygons, which we call~$P_S$.  We know that
  $P_S$ has at least five sides, one for each neighbor.  Each of the
  two \emph{bounding lines} contains a side to realize the adjacency
  to the other two gray polygons. We call the sides of $P_S$ that
  realize the adjacencies to the remaining polygons red, green, and
  blue, in correspondence to the vertex colors. The supporting line of
  the red side intersects each bounding line of~$S$.  We add a red
  point at each of the intersections.  For the blue and green sides we
  proceed analogously.  By convexity of $P_S$, these points are
  distinct.  This yields a permutation of red, green, blue (see
  \Cref{fig:3treeorder}) on each bounding line.  The permutations on
  the boundary of two adjacent strips coincide because the supporting
  lines are clearly contained in the supporting planes.

  \begin{figure}[tb]
    \begin{subfigure}[b]{0.22\textwidth}
      \centering \includegraphics[page=4]{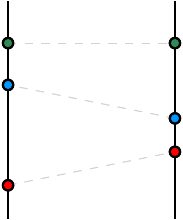}
      \caption{0 inversions}
      \label{fig:3treeorderA}
    \end{subfigure}
    \hfil
    \begin{subfigure}[b]{0.22\textwidth}
      \centering \includegraphics[page=2]{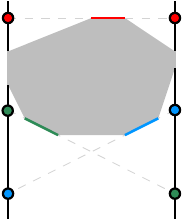}
      \caption{1 inversion}
      \label{fig:3treeorderB}
    \end{subfigure}
    \hfil
    \begin{subfigure}[b]{0.22\textwidth}
      \centering \includegraphics[page=5]{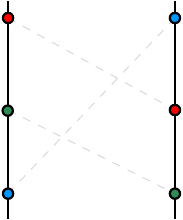}
      \caption{2 inversions}
      \label{fig:3treeorderC}
    \end{subfigure}
    \hfil
    \begin{subfigure}[b]{0.22\textwidth}
      \centering \includegraphics[page=3]{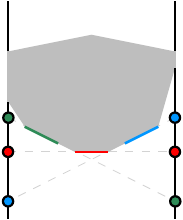}
      \caption{3 inversions}
      \label{fig:3treeorderD}
    \end{subfigure}
    \caption{The permutations of the intersections with the supporting
      lines of the red, green, and blue edges as in the proof of
      \Cref{lem:nonrealizable_3treeB}. Figures~(a) and~(c) illustrate
      possible scenarios.  Figures~(b) and~(d) show impossible
      scenarios because they do not contain cells of complexity 5.}
    \label{fig:3treeorder}
  \end{figure}

  Consider the line arrangement inside~$S$ given by the supporting
  lines of the red, green, and blue sides.  Up to symmetry,
  \Cref{fig:3treeorder} illustrates the different intersection
  patterns. To realize all contacts, the polygon~$P_S$ has to lie
  inside a cell incident to all five lines, namely the two bounding
  lines and the three supporting lines. It is easy to observe that
  such a cell exists only if the permutation has exactly one or three
  inversions; see \Cref{fig:3treeorderB,fig:3treeorderD}. In
  particular, the number of inversions is odd.
  
  Following the cyclic order of the bounding lines around the prism,
  we record three odd numbers of inversions in the permutations before
  coming back to the start.  Since an odd number of inversions does
  not yield the identity, we obtain the desired contradiction.
\end{proof}
      
Together,
\cref{lem:characterizePlanarThreeTree,lem:nonrealizable_3treeB} yield
the following theorem.

\begin{theorem}
  \label{thm:threeTree}
  Let $G$ be a 3-tree.  There exists a convex-polyhedral surface \Sur
  in~$\mathbb{R}^3$ with $\Gra(\Sur) \simeq G$ if and only if $G$ is
  planar.
\end{theorem}

In contrast to \cref{thm:threeTree}, there are nonplanar 3-degenerate
graphs that can be realized; see the example in \Cref{fig:3deg}.

\subsection{Hypercubes}
\label{sec:hypercubes}
In a paper from 1983, McMullen, Schulz, and Wills construct a
polyhedron for every integer $p\ge 4$ such that all faces are convex
$p$-gons~\cite[Sect.~4]{msw-p2me3-IJM83}.  In the following, we show
and illustrate how their result proves the realizability of any
hypercube.

\begin{proposition}[\cite{msw-p2me3-IJM83}]
  \label{prop:hypercube}
  For every $d$-hypercube $Q_d$, $d\ge0$, there exists a
  convex-polyhedral surface~\Sur in~$\mathbb{R}^3$ with
  $\Gra(\Sur) \simeq Q_d$ and every polygon of~\Sur is a $(d+4)$-gon.
\end{proposition}

The main building block in their construction is a polyhedral surface
whose adjacency graph is a $(p-4)$-hypercube.  In fact, we observed
that the adjacency graph of the polyhedron they finally construct is
the Cartesian product of $Q_{p-4}$ and a cycle graph $C_n, n\geq 3$.
For the first few steps of their inductive construction; see
\Cref{fig:mcmullen}.

\begin{figure}[ht]
  \begin{subfigure}[b]{.3\textwidth}
    \centering \includegraphics[height=4.5cm]{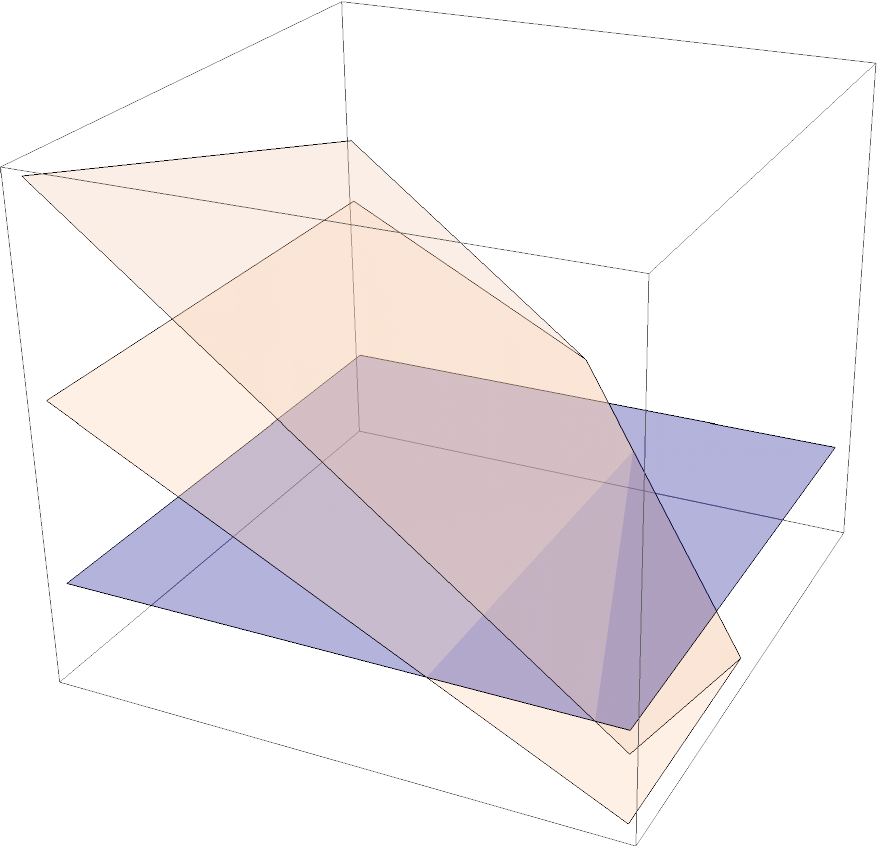}
    \caption{realization of $Q_1$ (orange), xy-plane (purle)}
    \label{fig:mcmullen1}
  \end{subfigure}
  \hfill
  \begin{subfigure}[b]{.2\textwidth}
    \centering \includegraphics[height=4.5cm]{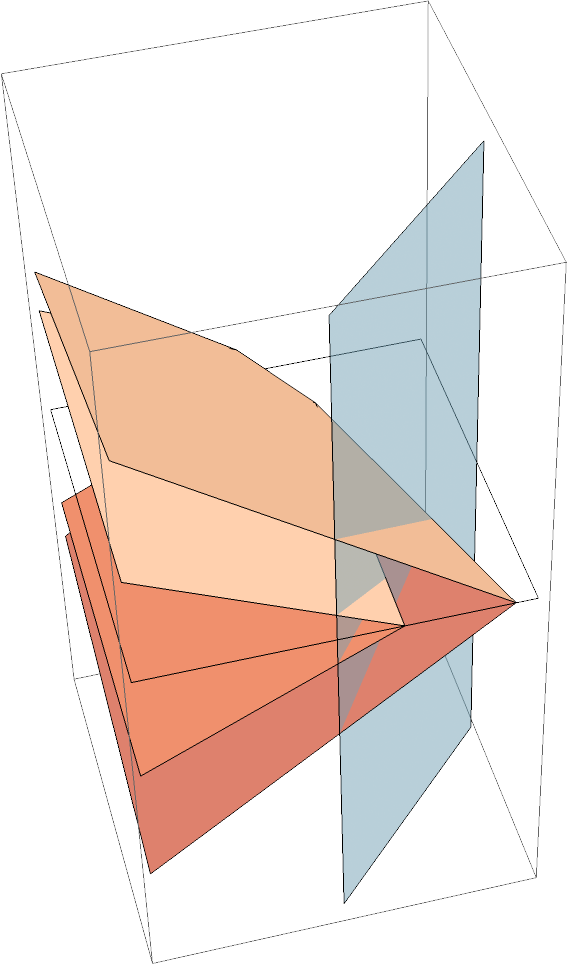}
    \caption{inductive step, cutting plane blue}
    \label{fig:mcmullen2}
  \end{subfigure}
  \hfill
  \begin{subfigure}[b]{.28\textwidth}
    \centering \includegraphics[height=4.5cm]{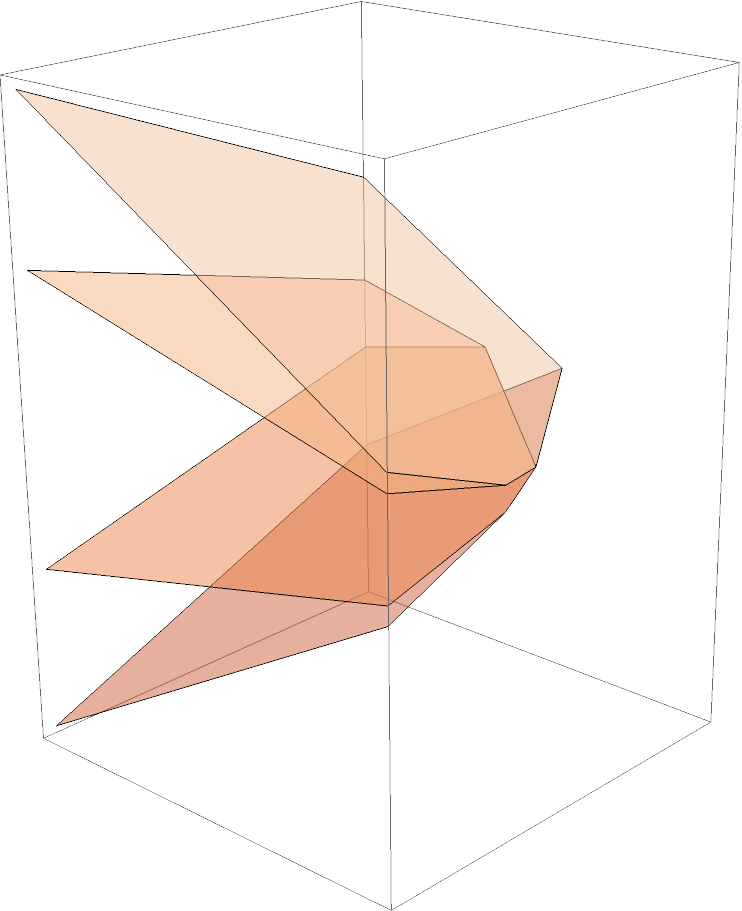}
    \caption{realization of $Q_2$}
    \label{fig:mcmullen3}
  \end{subfigure}
  \caption{The inductive construction of McMullen et
    al.~\cite{msw-p2me3-IJM83}}
  \label{fig:mcmullen}
\end{figure}

Recall that the $d$-hypercube has $2^d$ vertices.  The base case for
$d=0$ is given by a single $4$-gon, namely by the unit square.  What
follows is a series of inductive steps. In every step, the value of
$d$ increases by one and the number of polygons doubles. Before
explaining the step, we state the invariants of the construction. We
label the corners of a polygon with $p_1,\ldots, p_k$.  After every
step, the orthogonal projection into the xy-plane looks like the unit
square in which we have replaced the upper right corner with a convex
chain as shown in \Cref{fig:mcmullen2d}(a). In particular $p_1$ is
mapped to $(0,0)$, $p_2$ is mapped to $(0,1)$ and $p_k$ is mapped to
$(1,0)$.  For every polygon the sides $p_ip_{i+1}$ for
$i\le 3 \le k-2$ (non-vertical, non-horizontal in the projection) will
already have two incident polygons, the four other sides
$p_1p_2,p_2p_3,p_{k-1}p_k,p_kp_1$ are currently incident to only one
polygon.

\begin{figure}[htb]
  \centering
  \begin{tabular}{@{}cp{.01cm}cp{.01cm}c@{}}
    \includegraphics[page=1]{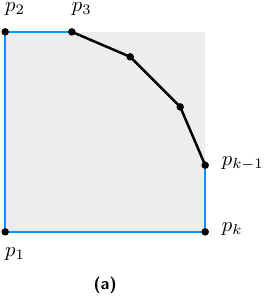}
    && \includegraphics[page=2]{mcmullen2d}
    && \includegraphics[page=3]{mcmullen2d}
  \end{tabular}
  
  \caption{Projection of a polygon into the xy-plane in the
    construction of McMullen et al.  The gray rectangle depicts the
    unit square.  Edges incident to only one polygon are drawn in
    blue. (a)~The start configuration for $d+4=7$.  (b)~Cutting with
    the xy-plane after the shear that puts only $e$ below the xy-plane
    (before glueing the reflected copy). (c)~Slicing off a corner to
    get the initial situation for $d+4=8$ modulo a projective
    transformation.}
  \label{fig:mcmullen2d}
\end{figure}

We explain next how to execute the inductive step. Suppose that we
have a polyhedral surface where every polygon is a $(d+4)$-gon
fulfilling our invariant. We apply a shear along the z-axis to assure
that for every polygon the corners $p_{k-1}$ and $p_k$ have smaller
z-coordinates than every corner of a polygon that is not $p_{k-1}$ or
$p_k$.  We then shift the whole surface such that exactly the sides
$p_{k-1}p_k$ lie completely below the xy-plane; see also
\Cref{fig:mcmullen1}.  These transformations do not change the
projections of the polygons into the xy-plane. We then cut the surface
with the xy-plane and only keep the upper part.  By this we slice away
one of the sides in all polygons but also add a side that lies in the
xy-plane; see \Cref{fig:mcmullen2d}(b).  Each polygon now has a side
$p'_{k-1}p'_k$ that lies in the xy-plane and is disjoint from all
other polygons.  We now take a copy of the surface at hand and reflect
it across the xy-plane. Every polygon of the original (unreflected)
surfaces is now glued to its reflected copy via the common side in the
xy-plane.  With this step, we already have transformed the adjacency
graph from a $d$-hypercube to a $(d+1)$-hypercube. We only need to
bring the surface back into the shape required by the invariant. To do
so, we cut off a corner in every polygon (see
\Cref{fig:mcmullen2d}(c)) by slicing the whole construction with an
appropriate plane orthogonal to the xy-plane; see also
\Cref{fig:mcmullen2}.  This turns all $(d+4)$-gons into $(d+5)$-gons.
In particular, in every polygon we cut off $p'_k$ and add two corners
$p''_k$ and $p''_{k+1}$ as shown in~\Cref{fig:mcmullen2d}(c).
Finally, we apply a projective transformation to assure that the
invariant holds in the end of the induction step.  Such a
transformation can be obtained as follows. Assume that the line
connecting $p''_k$ and $p''_{k+1}$ in the xy-plane, has the form
$x=ay+b$, for some parameter $a$ and $b$.  Then the transformation is
given by
\[ (x,y,z) \mapsto \frac{1}{ay+b} (x, y(a+b),z). \]
It can be observed that this mapping leaves the projection of the
points $p_1,p_2$ into the xy-plane for every polygon
stationary. Moreover, for every polygon, the line containing $p''_k$
and $p''_{k+1}$ will be mapped to the line $x=1$ and the line
containing of $p_2$ and $p_3$ will be mapped to the line $y=1$ when
projected into the xy-plane.
Figures~\ref{fig:mcmullen1}--\ref{fig:mcmullen3} show spatial images
of this construction.

\subparagraph{Connection to a problem of coloring adjacency graphs}
Thomassen~\cite[page~98, Problem~2]{thomassenColorCriticalGraphs}
asked whether the adjacency graph of a polyhedral surface in
$\mathbb R^3$ which is homeomorphic to $S_g$ (an orientable surface of
genus $g$) has chromatic number bounded by some absolute constant. A
typical approach for proving this is to show that the graph has
bounded average degree. %
However, \cref{prop:hypercube} shows that the average degree is
unbounded, so another approach is needed.

\section{Bounds on the Density}
\label{sec:density-bounds}

It is an intriguing question how dense adjacency graphs of
convex-polyhedral surfaces can be.  In this section, we use
realizability and non-realizability results from the previous sections
to derive asymptotic bounds on the maximum density of such graphs,
which we phrase in terms of the relation between their number of
vertices and edges.

Let $\mathcal G_n$ be the class of graphs on $n$ vertices with a
realization as a convex-polyhedral surface in~$\mathbb{R}^3$.
Further, let $e_{\max}(n)=\max_{G \in \mathcal G_n} |E(G)|$ be the
maximum number of edges that a graph in $\mathcal G_n$ can have.

\begin{corollary}
  \label{cor:density}
  For any positive integer $n$, $e_{\max}(n)\in \Omega(n \log n)$ and
  $e_{\max}(n)\in \bigO(n^{9/5})$.
\end{corollary}

\begin{proof}
  For the lower bound, note that by \Cref{prop:hypercube}, every
  hypercube is the adjacency graph of a convex-polyhedral surface.  As
  the $d$-dimensional hypercube has $2^d$ vertices and $2^d \cdot d/2$
  edges, the bound follows.

  For the upper bound, we use that, by \Cref{thm:k581}, the adjacency
  graph of a convex-polyhedral surface cannot contain $K_{5,81}$ as a
  subgraph.  It remains to apply the K\H{o}vari--S\'os--Tur\'an
  Theorem~\cite{KST54}, which states that an $n$-vertex graph that has
  no $K_{s,t}$ as a subgraph can have at most $\bigO(n^{2-1/s})$
  edges.
\end{proof}

Before being aware of the result of McMullen et
al.~\cite{msw-p2me3-IJM83}, we constructed a family of surfaces with
(large, but) constant average degree.  Our construction is not
recursive and therefore easier to understand and visualize; for a
sketch see \Cref{fig:6n}, a detailed description can be found in a
preprint version of this
article~\cite[Appendix~C]{akklsvw-dgps-arXiv21}.  Note that some
polygons in our construction have polynomial degree.

\begin{proposition}
  \label{thm:6n}
  There is an unbounded family of convex-polyhedral surfaces
  in~$\mathbb{R}^3$ whose adjacency graphs have average vertex degree
  $12-o(1)$.
\end{proposition}

\begin{figure}[tb]
  \begin{subfigure}[b]{0.38\textwidth}
    \centering \includegraphics[width=\textwidth]{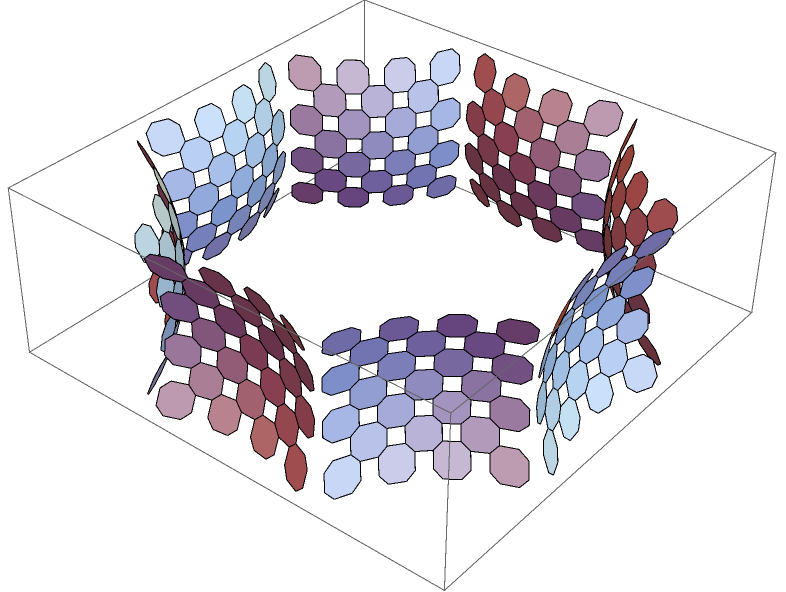}
    \caption{placement of octagon grids}
    \label{subfig:placement_gamma}
  \end{subfigure}
  \hfill
  \begin{subfigure}[b]{0.22\textwidth}
    \centering \includegraphics[width=.92\textwidth]{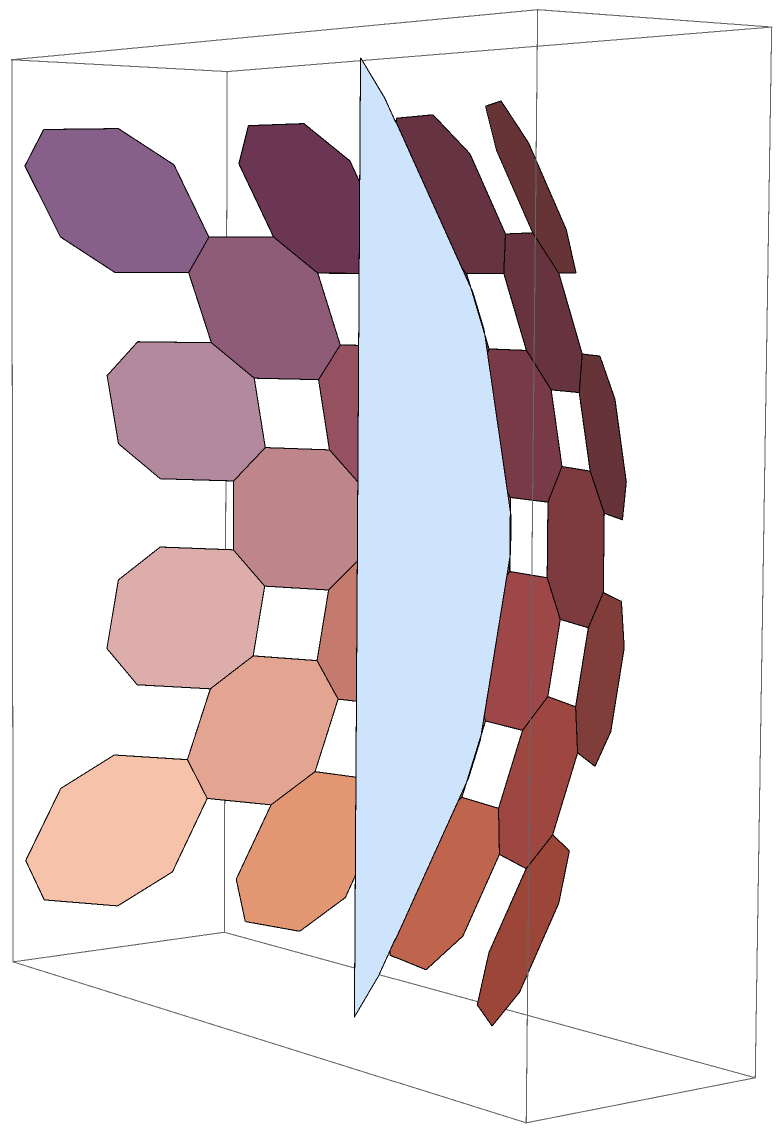}
    \caption{a vertical polygon}
    \label{subfig:vertical_polygon}
  \end{subfigure}
  \hfill
  \begin{subfigure}[b]{0.35\textwidth}
    \centering \includegraphics[width=.85\textwidth]{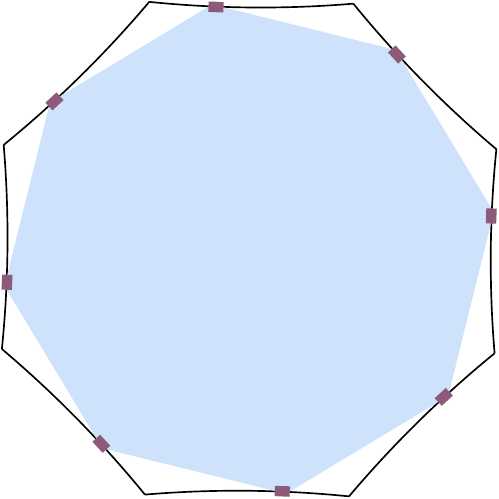}
    \caption{a horizontal polygon from above}
    \label{subfig:horizontal_polygons}
  \end{subfigure}
  \caption{A family of convex-polyhedral surfaces whose adjacency
    graphs have average vertex degree $12-o(1)$.  The main building
    block of our construction consists of $m$ regular octagons
    arranged in a truncated square tiling, which is lifted to the
    paraboloid; see (b).  We place $m$ copies of this gadget in a
    cyclic fashion; see (a).  To increase the average degree, we
    create $\bigO(m\sqrt m)$ vertical and $\bigO(m)$ horizontal
    polygons.  The vertical polygons are attached to the ``outside''
    of the bent octagon grids (see (b); the horizontal polygons are
    place in the center of our construction such that each of them
    touches each grid along a single polygon side (see (c)).  The
    resulting construction contains some unwanted overlaps and
    intersections, which can be removed by modifying the initial grid
    structure slightly.}
  \label{fig:6n}
\end{figure}

\section{Conclusion and Open Problems}

In this paper, we have studied the class of graphs that can be
realized as adjacency graphs of (convex-)polyhedral surfaces.
\cref{cor:density} bounds the maximum number $e_{\max}(n)$ of edges in
realizable graphs with $n$ vertices by $\Omega(n \log n)$ and
$\bigO(n^{9/5})$.  It would be interesting to improve upon these
bounds.
\begin{question}
  What is the maximum number $e_{\max}(n)$ of adjacencies that a
  convex-polyhedral surface with $n$ polygons can have?
\end{question}
We conjecture that realizability is \NP-hard to decide.
\begin{question}
  What is the computational complexity to decide for a given graph~$G$
  whether there exists a convex-polyhedral surface~\Sur such that
  $\Gra(\Sur) \simeq G$?
\end{question}
The following question is related to the previous question regarding
recognition.
\begin{question}
  Which structural properties are necessary or sufficient for
  admitting side-contact representations with convex polygons
  in~$\mathbb{R}^3$?
\end{question}

\subparagraph{Data availability statement.}

Data sharing is not applicable to this article as no datasets were
generated or analyzed during the current study.

\bibliographystyle{plainurl}
\bibliography{abbrv,contacts}

\end{document}